\documentclass[submission,copyright,creativecommons]{eptcs}
\providecommand{\event}{QPL 2021}

\usepackage{breakurl}
\usepackage{underscore}
\usepackage{amsmath,amsthm}
\usepackage{physics}
\usepackage{quiver}
\usepackage{graphicx}
\usepackage{enumerate}

\newcommand{\dffack}{Supported by \emph{DFF $\mid$ Natural Sciences} International Postdoctoral Fellowship 0131-00025B}
\newcommand{\epsrcack}{Supported by EPSRC Fellowship EP/R044759/1}

\newcommand{\cat}[1]{\ensuremath{\mathbf{#1}}}
\newcommand{\PInj}{\cat{PInj}}
\newcommand{\Pfn}{\cat{Pfn}}
\newcommand{\Isometry}{\cat{Isometry}}
\newcommand{\Unitary}{\cat{Unitary}}
\newcommand{\CPTP}{\cat{CPTP}}

\newcommand{\C}{\cat{C}}
\newcommand{\D}{\cat{D}}

\DeclareMathOperator{\Ext}{Ext}

\newcommand{\Lb}{\ensuremath{\mathrm{Aux}}}
\newcommand{\Rb}{\ensuremath{\mathrm{Inp}}}
\newcommand{\LRb}[1]{\Lb_\otimes(\Rb_\oplus(#1))}
\newcommand{\Inv}{\ensuremath{\mathrm{Inv}}}
\newcommand{\id}{\ensuremath{\mathrm{id}}}
\newcommand{\opp}{\ensuremath{\mathrm{op}}}

\newcommand{\tot}{\xrightarrow}
\newcommand{\fromt}{\xleftarrow}
\newcommand{\ridm}[1]{\ensuremath{\overline{#1}}}

\newtheorem{theorem}{Theorem}
\newtheorem{proposition}[theorem]{Proposition}
\newtheorem{lemma}[theorem]{Lemma}
\newtheorem{corollary}[theorem]{Corollary}
\theoremstyle{definition}
\newtheorem{definition}[theorem]{Definition}

\newtheorem{remark}[theorem]{Remark}
\makeatletter
\newtheorem*{rep@theorem}{\rep@title}
\newcommand{\newreptheorem}[2]{%
\newenvironment{rep#1}[1]{%
 \def\rep@title{#2 \ref{##1}}%
 \begin{rep@theorem}}%
 {\end{rep@theorem}}}
\makeatother
\newreptheorem{proposition}{Proposition}

\newif\ifqpl
\newif\ifarxiv

\qpltrue

\title{Bennett and Stinespring, Together at Last}
\author{Chris Heunen\footnote{\epsrcack} \quad\quad Robin Kaarsgaard\footnote{\dffack}
\institute{University of Edinburgh}
\email{\quad chris.heunen@ed.ac.uk \quad\quad robin.kaarsgaard@ed.ac.uk}
}
\def\titlerunning{Bennett and Stinespring, Together at Last}
\def\authorrunning{C. Heunen \& R. Kaarsgaard}
\begin{document}
\maketitle

%!TEX root = qpl.tex

\begin{abstract}
  We present a universal construction that relates reversible dynamics on open systems to arbitrary dynamics on closed systems: the restriction affine completion of a monoidal restriction category quotiented by well-pointedness.
  This categorical completion encompasses both quantum channels, via Stinespring dilation, and classical computing, via Bennett's method.
  Moreover, in these two cases, we show how our construction can be essentially `undone' by a further universal construction.
  This shows how both mixed quantum theory and classical computation rest on entirely reversible foundations.
\end{abstract}

\section{Introduction}
\label{sec:introduction}

Two constructions relate reversible dynamics on open systems to arbitrary dynamics on closed systems:
\begin{itemize}
  \item \emph{Stinespring dilation} realises a quantum channel as a reversible process on a larger space~\cite{stinespring:positive}.
  \item \emph{Bennett's method} makes a classical computer program reversible by allowing extra output~\cite{bennett:logicalrev,soekenetal:embedding}.
\end{itemize}
This paper presents a universal categorical construction encompassing both, making precise how the relationship between pure and mixed quantum theory resembles the relationship between reversible and conventional classical computation.

The construction has three phases: allowing additional constant input, leakage of output, and making it extensional. 
The first two phases adjoin auxiliary systems to the processes in question. The ancilla input can be seen as a form of temporary storage, while the output ancilla is not considered part of the desired output, and therefore is sometimes called \emph{garbage}. However, the garbage cannot be discarded without altering the function.
The third phase of the construction ensures that at least the garbage is \emph{extensional} (specific to the \emph{map} being computed rather than the \emph{method} used to compute it), so that equality of morphisms is judged solely on their observable input-output behaviour.
% Both phases have universal properties. 

% \begin{center}
%   \begin{tabular}{lll}
%     Category & result of the construction $\Ext(\Lb(-))$ \\
%     \hline 
%     partial injections between sets & partial functions \\
%     isometries of finite-dimensional Hilbert spaces & completely positive trace-preserving maps
%   \end{tabular}  
% \end{center}

We can also go in the converse direction by taking the cofree inverse category. All four phases have universal properties. On the whole, this shows how both mixed quantum theory and classical computation rest on entirely reversible foundations.

% \begin{center}
%   \begin{tabular}{lll}
%     Category & result of the construction $\Inv(\Rb(-))$ \\
%     \hline 
%     partial functions & partial injections \\
%     completely positive trace-preserving maps & unitaries
%   \end{tabular}  
% \end{center}

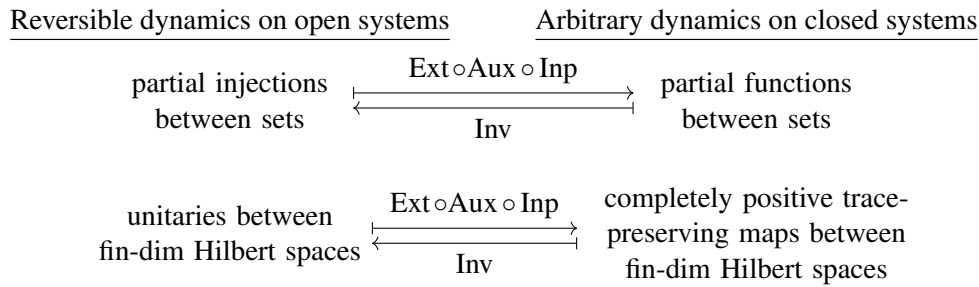
\[\begin{tikzpicture}
  \node at (0,2.3) {\underline{Reversible dynamics on open systems}};
  \node at (7,2.3) {\underline{Arbitrary dynamics on closed systems}};
  \node[text width=3cm,align=center] (tl) at (0,1.3) {partial injections between sets};
  \node[text width=3cm,align=center] (tr) at (7,1.3) {partial functions between sets};
  \node[text width=45mm,align=center] (r) at (7,-0.5) {completely positive trace-preserving maps between fin-dim Hilbert spaces};
  \node[text width=35mm,align=center] (bl) at (0,-0.5) {unitaries between fin-dim Hilbert spaces};
  \draw[|->] ([yshift=1mm]tl.east) to node[above]{$\Ext \circ \Lb \circ \Rb$} ([yshift=1mm]tr.west);
  \draw[|->] ([yshift=-1mm]tr.west) to node[below]{$\Inv$} ([yshift=-1mm]tl.east);
  \draw[|->] ([yshift=1mm]bl.east) to node[above]{$\Ext \circ \Lb \circ \Rb$} ([yshift=1mm]r.west);
  \draw[|->] ([yshift=-1mm]r.west) to node[below]{$\Inv$} ([yshift=-1mm]bl.east);
\end{tikzpicture}\]

\iffalse
\[\begin{tikzpicture}
  \node at (0,2.3) {\underline{Reversible dynamics on open systems}};
  \node at (7,2.3) {\underline{Arbitrary dynamics on closed systems}};
  \node[text width=3cm,align=center] (tl) at (0,1.3) {partial injections between sets};
  \node[text width=3cm,align=center] (tr) at (7,1.3) {partial functions between sets};
  \node[text width=35mm,align=center] (l) at (0,0) {isometries between fin-dim Hilbert spaces};
  \node[text width=45mm,align=center] (r) at (7,-0.65) {completely positive trace-preserving maps between fin-dim C*-algebras};
  \node[text width=35mm,align=center] (bl) at (0,-1.3) {unitaries between fin-dim Hilbert spaces};
  \draw[|->] ([yshift=1mm]tl.east) to node[above]{$\Ext \circ \Lb$} ([yshift=1mm]tr.west);
  \draw[|->] (l) to node[above]{$\Ext \circ \Lb$} (r);
  \draw[|->] ([yshift=-1mm]tr.west) to node[below]{$\Rb \circ \Inv$} ([yshift=-1mm]tl.east);
  \draw[|->] (r) to node[below]{$\Rb \circ \Inv$} (bl);
\end{tikzpicture}\]
\fi

There are some idiosyncracies among the four phases.
The $\Inv$-construction recovers partial injections from partial functions exactly, but only recovers unitaries from completely positive trace-preserving maps up to a global phase.
%The $\Inv$-construction leaves the category of partial injections invariant, whereas it turns isometries into unitaries.
The $\Rb$-construction leaves the category of partial injections invariant,
whereas it turns unitaries into isometries.
The $\Ext$-construction leaves the category of completely positive trace-preserving maps invariant, because minimal Stinespring dilations exist.
That is, Stinespring dilation allows an extensional choice of auxiliary system, whereas reversibilising embeddings are intensional. There are several (canonical) methods to make irreversible programs reversible. For example, Bennett's method stores the input and returns it in full along with the output, while the \emph{Landauer
embedding}~\cite{axelsenglueck:whatdo,landauer:irreversibility} additionally returns a trace of all instructions and attendant intermediate states.

\textbf{Related work}\quad
Both Stinespring dilation and Bennett's method have seen categorical presentations. Despite the similarity of their statements, these categorical completions are surprisingly dissimilar.
The universal construction of completely positive trace-preserving maps from isometries and unitaries is due to Huot and Staton~\cite{huotstaton:universal,huotstaton:completion}.
A different categorical approach to Stinespring's dilation
theorem as a universal construction is given by Westerbaan and
Westerbaan~\cite{westerbaan:paschke}.
The equivalence of discrete cartesian restriction categories and discrete inverse categories is due to Giles~\cite{giles:thesis}, though later recast by Comfort~\cite{comfort:zxamp} as a counital completion of inverse categories with chosen semi-Frobenius algebras.
Our $\Lb$-construction generalises a result by Hermida and Tennent ~\cite{hermidatennent:indeterminates}.
Combining it with our $\Ext$-construction gives the well-pointed completion of a monoidal restriction category that generalises both Huot-Staton and Giles.

\textbf{Future work}\quad
Following Giles, we conjecture that there is an equivalence between a category of certain monoidal inverse categories and certain well-pointed monoidal restriction categories.
Another interesting question is whether there is a minimal set that can be adjoined to any partial function to make it injective. Such a minimal Bennett embedding, as the miniminal Stinespring dilation, could be used to measure the degree to which a map is reversible. It may relate to the
information theoretic characterisation of reversible maps as those that
preserve entropy~\cite{landauer:irreversibility}.

\textbf{Overview}\quad
We assume familiarity with basic category theory. 
Section~\ref{sec:categorical-background} briefly recalls restriction categories and inverse categories. 
In Section~\ref{sec:the-bar-l-construction}, we present the $\Lb$-construction and show that it is the affine completion of a restriction monoidal category. 
Next, Section~\ref{sec:extensionality} introduces the $\Ext$-construction, and shows that it is governed by a universal property.
The constructions are put to work in Section~\ref{sec:examples} by showing that $\Ext \circ \Lb$ completes isometries to quantum channels and partial injective functions to partial functions. 
In Section~\ref{sec:cofree-reversible-foundations}, we use the \emph{dual} $\Rb$ of the $\Lb$-construction to show how quantum channels and partial functions can be universally constructed from unitaries and partial injections, respectively, and further that the latter can be recovered from the former by the $\Inv$-construction.
Appendix~\ref{sec:proofs} holds proofs that would distract in the main body of the article.
% section introduction (end)

\section{Restriction categories and inverse categories}
\label{sec:categorical-background}

While we assume basic familiarity with category theory, and in particular
monoidal categories~\cite{heunenvicary:cqt}, we briefly summarise restriction categories and inverse categories, which is relatively less well-known.
Restriction categories~\cite{cockettlack:restcat1} axiomatise partially defined morphisms. The idea is to record for each morphism $f$ its \emph{restriction idempotent} $\ridm{f}$, a partial identity defined precisely where $f$ is defined. 

\begin{definition}\label{def:restrictioncategory}
  A \emph{restriction category} is a category equipped with a choice of endomorphism $\ridm{f} \colon A \to A$ for each morphism $f \colon A \to B$ satisfying:
  \begin{enumerate}[(i)]
    \item $f \circ \ridm{f} = f$;
    \item $\ridm{f} \circ \ridm{g} = \ridm{g} \circ \ridm{f}$;
    \item $\ridm{g \circ \ridm{f}} = \ridm{g} \circ \ridm{f}$;
    \item $\ridm{g} \circ f = f \circ \ridm{g \circ f}$.
  \end{enumerate}
\end{definition}

The restriction idempotent $\ridm{f}$ measures `how partial' $f$ is. If $\ridm{f} = \id$, we call $f$ \emph{total}.
Any category becomes a restriction category when endowed with the trivial choice $\ridm{f}=\id$, but many other choices may be possible. 
When working with a restriction category, we often leave implicit which choice is made, just like the choice of tensor product making a category monoidal.
When we speak of the following categories, we will use the trivial restriction structure: \Unitary{} has finite-dimensional Hilbert spaces as objects and unitary linear maps as morphisms; \Isometry{} has finite-dimensional Hilbert spaces as objects and isometric linear maps as morphisms; \CPTP{} has finite-dimensional Hilbert spaces as objects and completely positive trace-preserving maps as morphisms.

But there are also nontrivial choices of restriction structure. On the category \Pfn{} of sets and partial functions, we will choose the restriction idempotent of a partial function $f \colon A \to B$ as follows:
\[
  \ridm{f}(x) = \begin{cases}
    x & \text{if $f$ is defined at $x$} \\
    \text{undefined} & \text{otherwise}
  \end{cases}
\]
Thus a partial function $f$ is total in the usual sense precisely when
it is total in the abstract sense.

A functor $F \colon \cat{C} \to \cat{D}$ between restriction categories is a \emph{restriction functor} when $F(\ridm{f}) = \ridm{F(f)}$. 
A (symmetric) \emph{monoidal restriction category} is a restriction category which that is also (symmetric) monoidal, such that the monoidal product is a restriction bifunctor: $\ridm{f \otimes g} = \ridm{f} \otimes \ridm{g}$. 

Similarly, restriction limits and colimits are ones that respect the restriction structure, though especially limits tend to be quite different. A \emph{restriction terminal object} is an object $1$ such that each object $A$ allows a unique \emph{total} morphism $A \to 1$. Restriction terminal objects need not be terminal in the usual sense; for
example, any singleton set is restriction terminal but \emph{not} terminal in \Pfn, because there is (at least) also the nowhere defined function $A \to 1$.

\begin{lemma}\cite{cockettlack:restcat1}
  For all appropriate $f$ and $g$ in a restriction category:
  \begin{enumerate}[(i)]
    \item $\ridm{g \circ f} = \ridm{\ridm{g} \circ f}$;
    \item $\ridm{g \circ f} = \ridm{f}$ if $g$ is total;
    \item $\ridm{f} = \id$ if $f$ is invertible.
  \end{enumerate}
\end{lemma}
% \begin{proof}
%   For \emph{(i)}: $\ridm{\ridm{g} \circ f} = \ridm{f \circ \ridm{g \circ 
%   f}} = \ridm{f} \circ \ridm{g \circ f} = \ridm{g \circ f} \circ \ridm{f} =\ridm{g \circ f \circ \ridm{f}} = \ridm{g \circ f}$. 
%   Now \emph{(ii)} follows from \emph{(i)}: $\ridm{g \circ f} = \ridm{\ridm{g} \circ f} = \ridm{\id \circ f} = \ridm{f}$.
%   Finally, for \emph{(iii)}: $\ridm{f} = \id \circ \ridm{f} = f^{-1}
%   \circ f \circ \ridm{f} = f^{-1} \circ f = \id$.
% \end{proof}

A morphism $f \colon A \to B$ in a restriction category is a
\emph{partial isomorphism} if there is a morphism $f^\circ \colon B \to A$ such that $f^\circ \circ f = \ridm{f}$ and $f \circ f^\circ = \ridm{f^\circ}$. Such partial inverses are unique whenever they exist. In \Pfn,
the partial isomorphisms are precisely the partial injective functions.

Recall that in a \emph{dagger category}, every morphism $f \colon A \to B$ has a partner $f^\dag \colon B \to A$ such that $f^{\dag\dag}=f$, $\id^\dag = \id$, and $(g \circ f)^\dag = f^\dag \circ g^\dag$~\cite{heunenvicary:cqt}. 
%Remarkably, restriction categories are connected to the theory of \emph{inverse categories}~\cite{kastl:inverse}. 

\begin{proposition}\cite{cockettlack:restcat1}
  The following are equivalent:
  \begin{enumerate}[(i)]
    \item \C{} is a restriction category in which each morphism is a partial isomorphism;
    \item \C{} is an \emph{inverse category}: a dagger category with $f \circ f^\dagger \circ f = f$ and $f^\dagger \circ f \circ g^\dagger \circ g = g^\dagger \circ g \circ f^\dagger \circ f$.
  \end{enumerate}
\end{proposition}

Inverse categories were originally conceived as a categorical extension of inverse semigroups~\cite{kastl:inverse}, but have recently seen applications as categorical models of classical reversible computation~\cite{giles:thesis, kaarsgaardaxelsenglueck:recursion, glueckkaarsgaard:flowcharts, gluecketal:revsem}. 
Examples of inverse categories include the category \PInj{} of sets and partial injective functions, as well as any groupoid (such as \Unitary). The connection between restriction and inverse categories generalises that between mere categories and groupoids.

\begin{proposition}\cite{kaarsgaardaxelsenglueck:recursion}\label{prop:cofreeinverse}
  The wide subcategory $\Inv(\C)$ of all partial isomorphisms of a (monoidal) restriction category $\cat{C}$ is its cofree (monoidal) inverse category: any inverse category $\cat{D}$ with a
  (strict monoidal) functor $\cat{D} \to \cat{C}$ allows a unique (strict
  monoidal) functor $\cat{D} \to \Inv(\cat{C})$ making the following diagram commute:
  % https://q.uiver.app/?q=WzAsMyxbMCwwLCJcXG1hdGhiZntEfSJdLFswLDEsIlxcbWF0aHJte0ludn0oXFxtYXRoYmZ7Q30pIl0sWzEsMSwiXFxtYXRoYmZ7Q30iXSxbMSwyXSxbMCwxLCIiLDAseyJzdHlsZSI6eyJib2R5Ijp7Im5hbWUiOiJkb3R0ZWQifX19XSxbMCwyXV0=
  \[\begin{tikzcd}
  	{\mathbf{D}} \\
  	{\mathrm{Inv}(\mathbf{C})} & {\mathbf{C}}
  	\arrow[from=2-1, to=2-2]
  	\arrow[dotted, from=1-1, to=2-1]
  	\arrow[from=1-1, to=2-2]
  \end{tikzcd}\]
  % Here, $\Inv(\C) \to \C$ is the (faithful) inclusion functor of $\Inv(\C)$ in $\C$.
\end{proposition}
If $\C$ in the above is a trivial restriction category, then $\Inv(\C)$ is its \emph{core}, that is, its cofree groupoid.
% section categorical_background (end)

\section{The \Lb-construction}
\label{sec:the-bar-l-construction}

This section is dedicated to the $\Lb$-construction, a generalisation of Hermida and Tennent's construction~\cite{hermidatennent:indeterminates} to (symmetric monoidal) restriction categories. After introducing $\Lb(\C)$, we show step by step that it is an affine monoidal restriction category.
Here, a monoidal restriction category is \emph{affine} when its tensor unit $I$ is restriction terminal. The crowning theorem shows that $\Lb(\C)$ is in fact the restriction affine completion of $\C$.

\begin{definition}
  Define a relation $\triangleright$ on the morphisms of a symmetric monoidal restriction category as follows. 
  For $f \colon A \to B \otimes E$ and $f' \colon A \to B \otimes E'$, set $f \triangleright f'$ if and only if $\ridm{f} = \ridm{f'}$ and there is a \emph{mediator} $h \colon E \to E'$ making the triangle commute:
  % https://q.uiver.app/?q=WzAsMyxbMSwwLCJBIl0sWzIsMSwiQiBcXG90aW1lcyBFJyJdLFswLDEsIkIgXFxvdGltZXMgRSJdLFswLDEsImYnIl0sWzAsMiwiZiIsMl0sWzIsMSwiXFxpZCBcXG90aW1lcyBoIiwyXV0=
  \[\begin{tikzcd}[row sep=5mm]
  	& A \\
  	{B \otimes E} && {B \otimes E'}
  	\arrow["{f'}", from=1-2, to=2-3]
  	\arrow["f"', from=1-2, to=2-1]
  	\arrow["{\id \otimes h}"', from=2-1, to=2-3]
  \end{tikzcd}\]
\end{definition}

This is a preorder: reflexivity follows by mediating with identities; transitivity follows by composing mediators. 
However, the relation need not be symmetric, for example if $\dim(E) < \dim(E')$ in \Isometry{}. 

\begin{definition}\label{def:equiv}
  Write $\sim$ for the equivalence relation generated by $\triangleright$. Explicitly, for $f \colon A \to B \otimes E$ and $f' \colon A \to B \otimes E'$, we have $f \sim f'$ if and only if there are intermediate morphisms $f_1, \dots, f_{n-1}$  with $\ridm{f} =
  \ridm{f_1} = \cdots = \ridm{f_{n-1}} = \ridm{f'}$ and mediators $E \tot{h_1} E_1 \fromt{h_2} E_2 \tot{h_3} \cdots \fromt{h_n} E'$ making the following diagram commute:
  % https://q.uiver.app/?q=WzAsNixbMCwxLCJCIFxcb3RpbWVzIEUiXSxbMSwxLCJCIFxcb3RpbWVzIEVfMSJdLFsyLDEsIkIgXFxvdGltZXMgRV8yIl0sWzMsMSwiXFxjZG90cyJdLFs0LDEsIkIgXFxvdGltZXMgRSciXSxbMiwwLCJBIl0sWzAsMSwiXFxpZCBcXG90aW1lcyBoXzEiLDJdLFsyLDEsIlxcaWQgXFxvdGltZXMgaF8yIl0sWzIsMywiXFxpZCBcXG90aW1lcyBoXzMiLDJdLFs0LDMsIlxcaWQgXFxvdGltZXMgaF9uIl0sWzUsNCwiZiciXSxbNSwwLCJmIiwyXSxbNSwyLCJmXzIgXFxxdWFkIFxcZG90cyJdLFs1LDEsImZfMSJdXQ==
  \[\begin{tikzcd}
  	&& A \\
  	{B \otimes E} & {B \otimes E_1} & {B \otimes E_2} & \cdots & {B \otimes E'}
  	\arrow["{\id \otimes h_1}"', from=2-1, to=2-2]
  	\arrow["{\id \otimes h_2}", from=2-3, to=2-2]
  	\arrow["{\id \otimes h_3}"', from=2-3, to=2-4]
  	\arrow["{\id \otimes h_n}", from=2-5, to=2-4]
  	\arrow["{f'}", from=1-3, to=2-5]
  	\arrow["f"', from=1-3, to=2-1]
  	\arrow["{f_2 \quad \dots}", from=1-3, to=2-3]
  	\arrow["{f_1}", from=1-3, to=2-2]
  \end{tikzcd}\]
\end{definition}

% We are finally ready to describe $\Lb(\C)$:

\begin{definition}
  For a symmetric monoidal restriction category \C, define a category
  $\Lb(\C)$:
  \begin{itemize}
    \item objects are those of \C;
    \item morphisms $[f,E] \colon A \to B$ are $\sim$-equivalence classes of morphisms $f : A \to B \otimes E$ in \C;
    \item composition of $[f,E] \colon A \to B$ and $[g,E'] \colon B \to C$ is $[\alpha \circ (g \otimes \id) \circ f, E' \otimes E] \colon A \to C$;
    \item identities are $[\rho^{-1},I] \colon A \to A$.
  \end{itemize}
\end{definition}

The previous definition differs from~\cite{hermidatennent:indeterminates} only by the additional requirement that $\ridm{f} = \ridm{f'}$ if $f \sim f'$.
It follows that the two are the same when \C{} is a trivial restriction category, making $\Lb$ a genuine generalisation.

\begin{remark}
  Morphisms in $\Lb(\C)$ are often given by composing chains of morphisms in 
  $\C$, further quotiented by a nontrivial equivalence relation.
  This can quickly become unintelligible.
  Therefore we will always make the zig-zag path of mediators from 
  Definition~\ref{def:equiv} explicit in equivalence arguments.
  To indicate which part of a diagram in \C{} corresponds to which morphism in 
  $\Lb(C)$, we will use squiggly grey `ghost' arrows:
  % https://q.uiver.app/?q=WzAsNCxbMCwwLCJBIl0sWzEsMCwiQiBcXG90aW1lcyBFIl0sWzIsMCwiKEMgXFxvdGltZXMgRScpIFxcb3RpbWVzIEUiXSxbMiwxLCJDIFxcb3RpbWVzIChFJyBcXG90aW1lcyBFKSJdLFswLDEsImYiXSxbMSwyLCJnwqBcXG90aW1lcyBcXGlkIl0sWzIsMywiXFxhbHBoYSJdLFswLDMsIihnLEUnKSBcXGNpcmMgKGYsRSkiLDIseyJjb2xvdXIiOlswLDAsNTBdLCJzdHlsZSI6eyJib2R5Ijp7Im5hbWUiOiJzcXVpZ2dseSJ9fX0sWzAsMCw1MCwxXV1d
  \[\begin{tikzcd}
  	A & {B \otimes E} & {(C \otimes E') \otimes E} \\
  	&& {C \otimes (E' \otimes E)}
  	\arrow["f", from=1-1, to=1-2]
  	\arrow["{g \otimes \id}", from=1-2, to=1-3]
  	\arrow["\alpha", from=1-3, to=2-3]
  	\arrow["{[g,E'] \circ [f,E]}"', color={rgb,255:red,128;green,128;blue,128}, squiggly, from=1-1, to=2-3]
  \end{tikzcd}\]
  This ghost arrow is \emph{not} a part of the commutative diagram. It merely indicates that $\alpha \circ (g \otimes \id) \circ f$
  corresponds precisely to $[g,E'] \circ [f,E]$ in $\Lb(\C)$.
\end{remark}

Notation settled, we now set out to show that this actually defines a restriction symmetric monoidal category. 
We proceed in three steps: first we show that it is a category; then that it inherits a restriction structure; and finally that it inherits a symmetric monoidal structure in a way that respects restriction. 
The proofs of the following three propositions are deferred to Appendix~\ref{sec:proofs} as they would distract from the main development.

\begin{proposition}\label{prop:LbCcat}
  $\Lb(\C)$ is a category.
\end{proposition}

\begin{proposition}\label{prop:LbCrestriction}
  $\Lb(\C)$ inherits a restriction structure from $\C$ with $\ridm{[f,E]} 
  = [\rho^{-1} \circ \ridm{f}, I]$.
\end{proposition}

\begin{proposition}\label{prop:LbCmonoidal}
  If \C{} is a restriction symmetric monoidal category, then so is $\Lb(\C)$:
  \begin{itemize}
    \item the tensor unit and tensor product of objects are as in $\C$;
    \item the tensor product of $[f,E] \colon A \to B$ and $[f',E'] \colon A' \to B'$ is $[\vartheta \circ (f \otimes f'), E \otimes E'] \colon A \otimes A' \to B \otimes B'$;
  \end{itemize} 
  where $\vartheta$ is the canonical isomorphism $(B \otimes E) \otimes (B' \otimes E') \simeq (B \otimes B') \otimes (E \otimes E')$ in \C.
\end{proposition}

Having established that $\Lb(\C)$ is a restriction symmetric monoidal category, our next goal is to show that it is the restriction affine completion of $\C$. 
Again we proceed in steps. 
First we show that there is a strict monoidal functor $\C \to \Lb(\C)$.
Then we show that the unit in $\Lb(\C)$ is restriction terminal, so that the tensor product has total projections. 
From this we derive a factorisation theorem for morphisms in $\Lb(\C)$, which finally lets us institute $\Lb(\C)$ as the restriction affine completion of \C.

\begin{proposition}
  If $\C$ is a restriction symmetric monoidal category, there is a strict 
  monoidal restriction functor $\mathcal{E} \colon \C \to \Lb(\C)$ given by $\mathcal{E}(A)=A$ on objects and by $\mathcal{E}(f) = [\rho^{-1} \circ f, I]$ on morphisms.
\end{proposition}
\begin{proof}
  To see $\mathcal{E}$ is functorial, compute $\mathcal{E}(\id) = [\rho^{-1} \circ
  \id, I] = [\rho^{-1}, I] = \id$. Composition is preserved because
  \begin{align*}
    \ridm{\mathcal{E}(g) \circ \mathcal{E}(f)} & = \ridm{\alpha \circ (\rho^{-1} \otimes \id) \circ 
    (g \otimes \id) \circ \rho^{-1} \circ f} = 
    \ridm{\alpha \circ (\rho^{-1} \otimes \id) \circ
    \rho^{-1} \circ g \circ f} = \ridm{g \circ f} \\
    & = \ridm{\rho^{-1} \circ g \circ f} = \ridm{\mathcal{E}(g \circ f)}
  \end{align*}
  and the diagram below commutes:
  % https://q.uiver.app/?q=WzAsMTAsWzAsMSwiQSJdLFsxLDAsIkIiXSxbNCwwLCIoQyBcXG90aW1lcyBJKSBcXG90aW1lcyBJIl0sWzUsMCwiQyBcXG90aW1lcyAoSSBcXG90aW1lcyBJKSJdLFs1LDEsIkMgXFxvdGltZXMgSSJdLFsxLDIsIkIiXSxbMiwyLCJDIl0sWzUsMiwiQyBcXG90aW1lcyBJIl0sWzIsMCwiQiBcXG90aW1lcyBJIl0sWzMsMCwiQyBcXG90aW1lcyBJIl0sWzAsMSwiZiJdLFsyLDMsIlxcYWxwaGEiXSxbMCw1LCJmIiwyXSxbNSw2LCJnIiwyXSxbNiw3LCJcXHJob157LTF9IiwyXSxbNyw0LCJcXGlkIFxcb3RpbWVzIFxcaWQiLDJdLFszLDQsIlxcaWQgXFxvdGltZXMgXFxyaG8iXSxbMCwzLCJGKGcpIFxcY2lyYyBGKGYpIiwyLHsiY29sb3VyIjpbMCwwLDUwXSwic3R5bGUiOnsiYm9keSI6eyJuYW1lIjoic3F1aWdnbHkifX19LFswLDAsNTAsMV1dLFswLDcsIkYoZyBcXGNpcmMgZikiLDAseyJjb2xvdXIiOlswLDAsNTBdLCJzdHlsZSI6eyJib2R5Ijp7Im5hbWUiOiJzcXVpZ2dseSJ9fX0sWzAsMCw1MCwxXV0sWzEsOCwiXFxyaG9eey0xfSJdLFs4LDksImcgXFxvdGltZXMgXFxpZCJdLFs5LDIsIlxccmhvXnstMX0gXFxvdGltZXMgXFxpZCJdXQ==
  \[\begin{tikzcd}
  	& B & {B \otimes I} & {C \otimes I} & {(C \otimes I) \otimes I} & {C \otimes (I \otimes I)} \\
  	A &&&&& {C \otimes I} \\
  	& B & C &&& {C \otimes I}
  	\arrow["f", from=2-1, to=1-2]
  	\arrow["\alpha", from=1-5, to=1-6]
  	\arrow["f"', from=2-1, to=3-2]
  	\arrow["g"', from=3-2, to=3-3]
  	\arrow["{\rho^{-1}}"', from=3-3, to=3-6]
  	\arrow["{\id \otimes \id}"', from=3-6, to=2-6]
  	\arrow["{\id \otimes \rho}", from=1-6, to=2-6]
  	\arrow["{F(g) \circ F(f)}"', color={rgb,255:red,128;green,128;blue,128}, squiggly, from=2-1, to=1-6]
  	\arrow["{F(g \circ f)}", color={rgb,255:red,128;green,128;blue,128}, squiggly, from=2-1, to=3-6]
  	\arrow["{\rho^{-1}}", from=1-2, to=1-3]
  	\arrow["{g \otimes \id}", from=1-3, to=1-4]
  	\arrow["{\rho^{-1} \otimes \id}", from=1-4, to=1-5]
  \end{tikzcd}\]
  The functor $\mathcal{E}$ preserves restriction idempotents: $\ridm{\mathcal{E}(f)} = [\rho^{-1} \circ \ridm{\rho^{-1} \circ f},I] = [\rho^{-1} \circ \ridm{f},I] = \mathcal{E}(\ridm{f})$. 
  That it is a strict monoidal functor follows from $\mathcal{E}(A \otimes B) = A \otimes B = \mathcal{E}(A) \otimes \mathcal{E}(B)$, $\mathcal{E}(I) = I$, $\mathcal{E}(f \otimes g) = \mathcal{E}(f) \otimes \mathcal{E}(g)$ (shown entirely analogously to showing $[\rho^{-1} \circ \beta,I] \otimes [\rho^{-1} \circ \phi,I] \sim [\rho^{-1} \circ (\beta \otimes \phi),I]$ for coherences $\beta$ and $\phi$ in Proposition~\ref{prop:LbCmonoidal}, see Appendix~\ref{sec:proofs}), and the fact that coherence isomorphisms in $\C$ are \emph{precisely} of the form $[\rho^{-1} \circ \beta, I] = \mathcal{E}(\beta)$ for each coherence isomorphism $\beta$ of $\C$.
\end{proof}

\begin{proposition}
  The tensor unit in $\Lb(\C)$ is restriction terminal.
\end{proposition}
\begin{proof}
  First note $I$ is weakly terminal: there is a morphism from each object $A$ into $I$, namely $[\lambda^{-1}, A]$.
  Furthermore, this morphism is total since $\ridm{[\lambda^{-1}, A]} = [\rho^{-1} \circ \ridm{\lambda^{-1}}, I] = [\rho^{-1} \circ \id, I]  = [\rho^{-1}, I] = \id$. 
  Because
  % https://q.uiver.app/?q=WzAsNSxbMCwxLCJBIl0sWzMsMCwiSSBcXG90aW1lcyBFIl0sWzMsMiwiSSBcXG90aW1lcyBBIl0sWzMsMSwiSSBcXG90aW1lcyAoSSBcXG90aW1lcyBFKSJdLFsyLDEsIkkgXFxvdGltZXMgRSJdLFswLDEsImYiLDAseyJjdXJ2ZSI6LTJ9XSxbMCwyLCJcXGxhbWJkYV57LTF9IiwyLHsiY3VydmUiOjJ9XSxbMSwzLCJcXGlkIFxcb3RpbWVzIFxcbGFtYmRhXnstMX0iXSxbMiwzLCJcXGlkIFxcb3RpbWVzIGYiLDJdLFs0LDMsIlxcbGFtYmRhXnstMX0iXSxbMCw0LCJmIl1d
  \[\begin{tikzcd}
  	&&& {I \otimes E} \\
  	A && {I \otimes E} & {I \otimes (I \otimes E)} \\
  	&&& {I \otimes A}
  	\arrow["f", curve={height=-12pt}, from=2-1, to=1-4]
  	\arrow["{\lambda^{-1}}"', curve={height=12pt}, from=2-1, to=3-4]
  	\arrow["{\id \otimes \lambda^{-1}}", from=1-4, to=2-4]
  	\arrow["{\id \otimes f}"', from=3-4, to=2-4]
  	\arrow["{\lambda^{-1}}", from=2-3, to=2-4]
  	\arrow["f", from=2-1, to=2-3]
  \end{tikzcd}\]
  any \emph{total} morphism $[f,E] \colon A \to I$ satisfies $[f,E] \sim [\lambda^{-1},A]$.
\end{proof}

We will simply write $!$ for the unique morphism $[\lambda^{-1},A] \colon A \to I$ from now on.

\begin{remark}
  An important property of restriction affine monoidal categories is that they have total maps $\pi_1 : A \otimes B \to A$ and $\pi_2 : A \otimes B \to B$. These can be defined as $A \otimes B 
  \tot{\id \otimes !} A \otimes I \tot{\rho} A$ and symmetrically, and are
  total since $\ridm{\rho \circ (\id \otimes !)} = \ridm{(\id \otimes !)} =
  \ridm{\id} \otimes \ridm{!} = \id \otimes \id = \id$, and similarly for 
  the second projection.
\end{remark}

These total projections are crucial in showing the following factorisation of morphisms in $\Lb(\C)$, based on Hermida and Tennent's \emph{expansion-raw morphism factorisation}~\cite[Lemma 2.8]{hermidatennent:indeterminates}.

\begin{lemma}\label{lem:factor}
  Every morphism $[f,E] : A \to B$ of $\Lb(\C)$ factors as $\pi_1 \circ \mathcal{E}(f)$. This factorisation is unique in the sense that if $[f,E] \sim \pi_1 \circ \mathcal{E}(f')$ for any $f'$, then $[f,E] \sim [f',E']$.
\end{lemma}
\begin{proof}
  Let $[f,E] : A \to B$ be a morphism of $\Lb(\C)$. First, 
  $\ridm{\pi_1 \circ \mathcal{E}(f)} = \ridm{\mathcal{E}(f)} =
  \mathcal{E}(\ridm{f}) = {[\rho^{-1} \circ \ridm{f}, I]} = \ridm{[f,E]}$.
  That $[f,E] \sim \pi_1 \circ \mathcal{E}(f)$ then follows by commutativity of 
  the diagram below.
  % https://q.uiver.app/?q=WzAsNixbMCwxLCJBIl0sWzEsMCwiQiBcXG90aW1lcyBFIl0sWzIsMCwiKEIgXFxvdGltZXMgRSkgXFxvdGltZXMgSSJdLFszLDAsIkIgXFxvdGltZXMgKEUgXFxvdGltZXMgSSkiXSxbMywyLCJCIFxcb3RpbWVzIEUiXSxbMywxLCJCIFxcb3RpbWVzIEUiXSxbMCwxLCJmIl0sWzEsMiwiXFxyaG9eey0xfSJdLFsyLDMsIlxcYWxwaGEiXSxbMCw0LCJmIiwyLHsiY3VydmUiOjJ9XSxbMyw1LCJcXGlkIFxcb3RpbWVzIFxccmhvIl0sWzQsNSwiXFxpZCBcXG90aW1lcyBcXGlkIiwyXSxbMCw0LCIoZixFKSIsMCx7ImNvbG91ciI6WzAsMCw1MF0sInN0eWxlIjp7ImJvZHkiOnsibmFtZSI6InNxdWlnZ2x5In19fSxbMCwwLDUwLDFdXSxbMCwzLCJcXHBpXzEgXFxjaXJjIFxcbWF0aGNhbHtFfShmKSIsMix7ImNvbG91ciI6WzAsMCw1MF0sInN0eWxlIjp7ImJvZHkiOnsibmFtZSI6InNxdWlnZ2x5In19fSxbMCwwLDUwLDFdXV0=
  \[\begin{tikzcd}
  	& {B \otimes E} & {(B \otimes E) \otimes I} & {B \otimes (E \otimes I)} \\
  	A &&& {B \otimes E} \\
  	&&& {B \otimes E}
  	\arrow["f", from=2-1, to=1-2]
  	\arrow["{\rho^{-1}}", from=1-2, to=1-3]
  	\arrow["\alpha", from=1-3, to=1-4]
  	\arrow["f"', curve={height=12pt}, from=2-1, to=3-4]
  	\arrow["{\id \otimes \rho}", from=1-4, to=2-4]
  	\arrow["{\id \otimes \id}"', from=3-4, to=2-4]
  	\arrow["{[f,E]}", color={rgb,255:red,128;green,128;blue,128}, squiggly, from=2-1, to=3-4]
  	\arrow["{\pi_1 \circ \mathcal{E}(f)}"', color={rgb,255:red,128;green,128;blue,128}, squiggly, from=2-1, to=1-4]
  \end{tikzcd}\]
  Now suppose $[f,E] \sim \pi_1 \circ \mathcal{E}(f')$ for some
  $f' : A \to B \otimes E'$ in \C. Similarly as before, $[f',E'] \sim \pi_1 \circ \mathcal{E}(f')$, so it simply follows by transitivity that $[f,E] \sim \pi_1 \circ \mathcal{E}(f') \sim [f',E']$.
\end{proof}

We have finally arrived at the main theorem of this section.

\begin{theorem}\label{thm:lbc-free}
  $\Lb(\C)$ is the restriction affine completion of a restriction
  symmetric monoidal category $\C$: given any other restriction affine
  symmetric monoidal category $\D$ and strong monoidal restriction functor $F \colon \C \to \D$, there is a unique functor $\hat{F} \colon \Lb(\C) \to \D$ with $F = \hat{F} \circ \mathcal{E}$.
  % https://q.uiver.app/?q=WzAsMyxbMCwwLCJcXG1hdGhiZntDfSJdLFsyLDAsIlxcYmFye0x9KFxcbWF0aGJme0N9KSJdLFsyLDEsIlxcbWF0aGJme0R9Il0sWzAsMSwiXFxtYXRoY2Fse0V9Il0sWzAsMiwiRiIsMl0sWzEsMiwiXFxoYXR7Rn0iXV0=
  \[\begin{tikzcd}
  	{\mathbf{C}} && {\Lb(\mathbf{C})} \\
  	&& {\mathbf{D}}
  	\arrow["{\mathcal{E}}", from=1-1, to=1-3]
  	\arrow["F"', from=1-1, to=2-3]
  	\arrow[dashed, "{\hat{F}}", from=1-3, to=2-3]
  \end{tikzcd}\]
\end{theorem}
\begin{proof}
  Define $\hat{F} \colon \Lb(\C) \to \D$ by $\hat{F}(A) = F(A)$ on objects, on a morphism $[f,E] \colon A \to B$ by:
  \[
    \hat{F}(A) = F(A) \tot{F(f)} F(B \otimes E) \tot{\mathrm{str}} F(B) \otimes F(E) \tot{\pi_1} F(B) = \hat{F}(B) 
  \]
  where $F(B \otimes E) \tot{\mathrm{str}} F(B) \otimes F(E)$ is the monoidal 
  strength.
  This makes the diagram commute since $\hat{F}(\mathcal{E}(A)) = \hat{F}(A) = 
  F(A)$ on objects, and on morphisms
  \[
    \hat{F}(\mathcal{E}(f)) = \hat{F}([\rho^{-1} \circ f,I]) = \pi_1 \circ \mathrm{str} \circ F(\rho^{-1} \circ f) = \pi_1 \circ \mathrm{str} \circ F(\rho^{-1}) \circ F(f) = F(f) 
  \]
  because 
  \[
    \pi_1 \circ \mathrm{str} \circ F(\rho^{-1}) = 
    \rho \circ (\id \otimes {!}) \circ \mathrm{str} \circ F(\rho^{-1}) = \rho 
    \circ \rho^{-1} = \id
  \]
  by definition of $\pi_1$ and right unitality of the monoidal strength.
  The functor $\hat{F}$ is strong monoidal because $F$ is, since $\hat{F}(A \otimes B) = F(A \otimes B) \simeq F(A) \otimes F(B)$ and since all coherence isomorphisms $\Lb(\C)$ are of the form $\mathcal{E}(\beta)$ for a coherence isomorphism $\beta$ of $\C$, so that $\hat{F}(\beta) = \hat{F}(\mathcal{E}(\beta)) = F(\beta) = \beta$.
  Also, $\hat{F}$ is a restriction functor since $F$ is:
  $\ridm{\hat{F}([f,E])} = \ridm{\pi_1 \circ F(f)} = \ridm{F(f)} = F(\ridm{f}) = \hat{F}(\mathcal{E}(\ridm{f})) = \hat{F}([\rho^{-1} \circ \ridm{f},I]) = \hat{F}(\ridm{[f,E]})$.
  
  To see that $\hat{F}$ is unique, suppose $G \colon \Lb(\C) \to \D$ is a strong monoidal restriction functor making the triangle commute. 
  First, $\hat{F}$ and $G$ agree on objects as $G(A) = \hat{F}(\mathcal{E}(A)) = \hat{F}(A)$. 
  If $[f,E] \colon A \to B$ is a morphism of $\Lb(\C)$, then Lemma~\ref{lem:factor} guarantees $[f,E] \sim \pi_1 \circ \mathcal{E}(f)$, so:
  \[
    G([f,E]) = G(\pi_1 \circ
  \mathcal{E}(f)) = G(\pi_1) \circ G(\mathcal{E}(f)) = \pi_1 \circ F(f) =
  \hat{F}([f,E])
  \qedhere
  \]
  % where $G(\pi_1) = \pi_1$ because $G$ is strict monoidal.
\end{proof}
% section the_bar_l_construction (end)

\section{Extensionality} % (fold)
\label{sec:extensionality}

Functional \emph{extensionality} means that two functions are equal if they return the same output on every input. This may not be the case in \emph{intensional} type theories.
This section concerns the second phase of our completion: the $\Ext$-construction. It quotients a given category by an equivalence relation related to well-pointedness to make it extensional, which we will show has a universal property. Combining this with the $\Lb$-construction of Section~\ref{sec:the-bar-l-construction}, the main results of this section will show that $\Ext(\Lb(\Isometry)) \simeq \CPTP$ and $\Ext(\Lb(\PInj)) \simeq \Pfn$. % This follows from proofs of $\Lb(\Isometry)/{\approx} \cong \CPTP$ and $\Lb(\PInj)/{\approx} \cong \Pfn$.

Say that a (restriction) category is \emph{pointed} if it has a (restriction) terminal object, and that it is \emph{(restriction) well-pointed} if additionally $f=g$ as soon as $f \circ a = g \circ a$ for all $a \colon 1 \to A$. Both \Pfn{} and \CPTP{} are restriction well-pointed.

\begin{definition}\label{def:cong}
  In a pointed restriction category, define a relation $\approx$ on parallel morphisms $f,g \colon A \to B$ by setting $f \approx g$ if and only if $f \circ a = g \circ a$ for all $a \colon 1 \to A$.
  Write $\Ext(\cat{C})$ for $\cat{C}\slash{\approx}$.
\end{definition}

\begin{lemma}
  The relation $\cdot \approx \cdot$ is a congruence, and so $\Ext(\cat{C})$ is a well-defined category.
\end{lemma}
\begin{proof}
  Suppose that $f,f' \colon A \to B$ and $g,g' \colon B \to C$ satisfy $f \approx f'$ and $g \approx g'$. Let $a \colon 1 \to A$. Then $f \circ a = f' \circ a$, and hence $g \circ f \circ a = g' \circ f' \circ a$. So $g \circ f \approx g' \circ f'$.
\end{proof}

The congruence $\approx$ also respects restriction structure: if $f,f' \colon A \to B$ satisfy $f \approx f'$, then also $\overline{f} \approx \overline{f'}$, by Definition~\ref{def:restrictioncategory}(iv), for if $a \colon 1 \to A$, then $\overline{f} \circ a = a \circ \overline{f \circ a} = a \circ \overline{f' \circ a} = \overline{f'} \circ a$.
Therefore $\Ext(\cat{C})$ is a well-defined restriction category, and the quotient functor $\cat{C} \to \Ext(\cat{C})$ sending a morphism to its equivalence class is a restriction functor.

However, it is not clear whether $\approx$ is a monoidal congruence when the category is affine monoidal. If $f \approx f' \colon A \to C$ and $g \approx g' \colon B \to D$, then $(f \otimes g) \circ x = (f' \otimes g') \circ x$ for all $x \colon 1 \to A \otimes B$ of the form $x=(a \otimes b) \circ \lambda_I^{-1}$ for $a \colon 1 \to A$ and $b \colon 1 \to B$. But what about entangled states $x \colon 1 \to A \otimes B$? Luckily, in the examples below this holds, so $\Ext(\cat{C})$ is again a well-defined monoidal category, and $\cat{C} \to \Ext(\cat{C})$ a strict monoidal functor.

By construction $\Ext(\C)$ is well-pointed, and the $\Ext$-construction is universal in accomplishing this.

\begin{definition}
  Call a functor $F \colon \cat{C} \to \cat{D}$ between pointed restriction categories \emph{full on points} if each $p \colon 1 \to F(A)$ in $\D$ is of the form $F(a)$ for some $a \colon 1 \to A$ in $\C$.
\end{definition}

\begin{theorem}\label{thm:well-pointed-completion}
  $\Ext(\C)$ is the \emph{well-pointed completion} of the pointed restriction category $\C$:
  given a well-pointed restriction category $\D$ and restriction
  functor $F \colon C \to D$ that is full on points, there is a unique restriction functor $\hat{F} \colon \Ext(\C) \to D$ that is full on points and makes the triangle commute:
  % https://q.uiver.app/?q=WzAsMyxbMCwwLCJcXG1hdGhiZntDfSJdLFsyLDAsIlxcbWF0aGJme0N9L3tcXGFwcHJveH0iXSxbMiwxLCJcXG1hdGhiZntEfSJdLFswLDFdLFswLDIsIkYiLDJdLFsxLDIsIlxcaGF0e0Z9Il1d
  \[\begin{tikzcd}
  	{\mathbf{C}} && {\Ext(\mathbf{C})} \\
  	&& {\mathbf{D}}
  	\arrow[from=1-1, to=1-3]
  	\arrow["F"', from=1-1, to=2-3]
  	\arrow[dashed,"{\hat{F}}", from=1-3, to=2-3]
  \end{tikzcd}\]
\end{theorem}
\begin{proof}
  Set $\hat{F}(A) = F(A)$ on objects and $\hat{F}([f]) = F(f)$ on morphisms. 
  To see that this is well-defined, suppose $f \approx g$, that is $f \circ a = g \circ a$ for all $a \colon 1 \to A$ in $\C$. 
  Then also $F(f) \circ p = F(g) \circ p$ for all $p \colon 1 \to F(A)$ in $\D$ because $F$ is full on points, and so $F(f) = F(g)$ since $\D$ is well-pointed. 

  Moreover, $\hat{F}$ is a restriction functor since $\ridm{\hat{F}([f])} = \ridm{F(f)} = F(\ridm{f}) = \hat{F}([\ridm{f}])$, and it is
  full on points since $F$ is.
  Now $\hat{F} \circ Q = F$ directly.
  It remains to show that $\hat{F}$ is the unique such functor. 
  Suppose $G \circ Q = F$ for a functor $G \colon \Ext(\C) \to \D$ that is full on points. But then $G(A) = F(A)$, and since $Q(f) = [f]$, we must also have $G([f]) = F(f) = \hat{F}([f])$.
\end{proof}

When the functor $\C \to \Ext(\C)$ is strict monoidal, as is the case
for both $\Isometry$ and $\PInj$, it completes restriction affine monoidal categories to restriction well-pointed monoidal categories. 

\section{Quantum channels and classical functions as completions}
\label{sec:examples}

This section instantiates the theory of the previous ones for our main examples.
The quantum case is quickly established thanks to Huot and Staton.

\begin{proposition}\label{prop:cptp}
  There is a monoidal equivalence $\Ext(\Lb(\Isometry)) \simeq \CPTP$.
\end{proposition}
\begin{proof}
  Since $\Isometry$ is a trivial restriction category, $\CPTP \simeq
  L(\Isometry) \simeq \Lb(\Isometry)$ by~\cite[Corollary~7]{huotstaton:universal}. 
  Also, $\CPTP$ is already well-pointed, so $\Ext(\CPTP) \simeq \CPTP$.
  It is easy to verify that the equivalence is monoidal.
\end{proof}

Above, the $\Ext$-phase was trivial, but this is not always the case. Consider the (intensional) category $\Lb(\PInj)$:
objects are sets, and morphisms $A \to B$ are partial injective functions $A \to B \times E$ that are identified $f \sim f'$ when there is a partial injective function $h \colon E \to E'$ such that $f(x) = (y,e)$ implies
$f'(x) = (y,h(e))$ for all $x \in A$. 
The environment $E$ is often thought of as the \emph{garbage} produced by the function because, being injective, it cannot actually discard any information. However, the $\Lb$-construction allows it to place instead the garbage off to the side, demarcating it from the desired output. In reversible computation, such garbage is unavoidable (since not all computable functions, and even not all interesting such, happen to be injective), so it is important that it is managed properly. 

Garbage is ideally extensional: we should be able to compare functions by looking only at their input-output behavior, even when some of it is designated as garbage. But unless you are careful, this might not be the case. 
Consider the successor function $n \mapsto n+1$ on natural numbers.
We can consider many different ways to vary the environment:
for example $f_1 \colon \mathbb{N} \to \mathbb{N} \times \{*\}$ given by $n \mapsto (n+1,*)$;
but also $f_2 \colon \mathbb{N} \to \mathbb{N} \times \mathbb{N}$ given by $n \mapsto (n+1,n)$. 
These two functions effect the exact same behaviour when disregarding garbage.
But they are in different equivalence classes as morphisms $\mathbb{N} \to \mathbb{N}$ in $\Lb(\PInj)$ because their garbage is so different.
% (since making them equivalent would require a partial injective function $\mathbb{N} \to \{\star\}$ capable of mapping \emph{any} natural number $n$ to $\star$, which is clearly an impossible task for an injective function).

How to mend this? First notice that points $1 \to A$ in $\Lb(\PInj)$ correspond to those in $\Pfn$ (see Lemma~\ref{lem:points-pfn} below). 
Even though $f_1$ and $f_2$ are different in $\Lb(\PInj)$, they do agree on each $n \colon 1 \to \mathbb{N}$: build the partial function $h_n \colon \mathbb{N} \to 1$ \emph{defined only on $n$} by $h_n(n) = *$; this mediates because $f_1(n) = (n+1,*) = (n+1,h_n(n))$ and $f_2(n) = (n+1,n)$. 
So, garbage is intensional in $\Lb(\PInj)$ because the category is not well-pointed.
%since $\mathsf{inc}' \circ n \sim \mathsf{inc}'' \circ n$ for all $1 \tot{n} \mathbb{N}$ yet $\mathsf{inc}' \neq \mathsf{inc}''$. 
Because \Pfn{} is well-pointed, it is necessary to identify morphisms when they agree on all points, which is exactly what the $\Ext$-construction does.

Why was this not an issue in the quantum case? There, extensionality arises from \emph{minimal} Stinespring dilations. Minimality gives a unique minimal (up to unitary) auxiliary system we can adjoin to realise any CPTP-map as conjugation by an isometry, thus taking away the choice of environment $E$ that sparked the trouble in $\Lb(\PInj)$.

\begin{lemma}\label{lem:points-pfn}
  The global points $1 \to A$ in $\Lb(\PInj)$ coincide with those in $\Pfn$.
\end{lemma}
\begin{proof}
  Points in $\Lb(\PInj)$ are partial injective functions $x \colon 1 \to A \times E$ modulo identification. However, any such point can always 
  be identified with one of the form $y \colon 1 \to A \times 1$ since if $x(*) = (a,e)$ then the point $* \mapsto e$ mediates $1 \to E$ to witness $(x,E) \sim (y,1)$. If $E$ is the empty set, the nowhere defined function trivially mediates.
\end{proof}

It follows from the previous Lemma that the functor $\Lb(\PInj) \to \Ext(\Lb(\PInj))$ is full on points. So is $\Lb(\Isometry) \to \Ext(\Lb(\Isometry))$, but in a trivial way: because \ifarxiv\linebreak\fi$\Lb(\Isometry) \simeq \CPTP$ by~\cite{huotstaton:universal}, and $\CPTP$ is already well-pointed, this functor is an isomorphism of categories. 

% As a consequence of this proposition, morphisms $(f,E)$ and $(f',E')$ are identified in $\Lb(\PInj)/{\approx}$ if for all $x \in A$ there exists a partial injective function $E \tot{h_x} E'$ such that $f(x) = (y,e)$ implies $f'(x) = (y,h_x(e))$. 

\begin{proposition}\label{prop:bennett}
  There is a monoidal equivalence $\Ext(\Lb(\PInj)) \simeq \Pfn$.
\end{proposition}
\begin{proof}
  Define $F \colon \Pfn \to \Ext(\Lb(\PInj))$ by $F(A) = A$ on objects, and on morphisms $f \colon A \to B$ by $F(f) = [b_f,A]$, where $b_f$ is the \emph{Bennett embedding} of $f$ given by $b_f(x) = (f(x),x)$. 
  
  We argue first that this is functorial: $F(\id)$ is $b_\id(x) = (x,x)$, but the chosen identity is (the equivalence class of) $\rho^{-1}(x) = (x,\star)$. However, on each point $p$, simply choose $p$ itself to
  mediate to see $b_\id \approx \rho^{-1}$. 
  Likewise, whereas $F(g \circ f)$ is $b_{g \circ f}(x) = (g(f(x)),x)$ and $F(g) \circ F(f)$ is $b'(x) = \big(g(f(x)), (f(x), x)\big)$, for each point $x$, mediate that point by $h_x \colon A \to B \times A$ given by:
  \[
    h_x(a) = \begin{cases}
      (f(x),x) & \text{if } a=x \\
      \text{undefined} & \text{otherwise}
    \end{cases}
  \]
  Thus $F(g \circ f) \approx F(g) \circ F(f)$. 
  Since $\Pfn$ and $\Ext(\Lb(\PInj))$ have the same objects, it remains only to be seen that $F$ is full and faithful. 

  For fullness, let a partial injective $f \colon A \to B \times E$ represent a morphism in $\Ext(\Lb(\PInj))$. 
  Since $[f,E]$ and $[f',E']$ are identified if and only if for all $x \in A$ there exists a partial injective function $h_x \colon E \to E'$ such that $f(x) = (y,e)$ implies $f'(x) = (y,h_x(e))$, either way $\pi_1 \circ f = \pi_1 \circ f'$ as partial functions.
  % So irrespective of the representative $f$, we can always associate it with a partial function in a well-defined way.
  Consider now the Bennett embedding of $\pi_1 \circ f$, that is, the partial injective function $b_{\pi_1 \circ f} \colon A \to B \times A$ given by $x \mapsto (\pi_1(f(x)), x)$, and compare it to $f \colon A \to B \times E$.
  For any $x \in A$, it follows that if $f(x) = (y,e)$ then 
  $b_{\pi_1 \circ f}(x) = (\pi_1(f(x)), x) = (\pi_1(y,e),x) = (y,x)$, so the two agree in the first component. Define a one-point mediator $h_x \colon A \to E$ for $x$ given by:
  \[
    h_x(a) = \begin{cases}
      e & \text{if } a=x \\
      \text{undefined} & \text{otherwise}
    \end{cases}
  \]
  Thus $b_{\pi_1 \circ f} \approx f$ and $F$ is full.

  Towards faithfulness, suppose $F(f) \approx F(g)$, so $b_f \approx b_g$ for some $f,g \colon A \to B$. Thus $b_f(x) = (f(x),x)$
  for some partial function $f$, and similarly $b_g(x) = (g(x),x)$. That $b_f \approx b_g$ means that for each $a \in A$ there exists $h_a : A \to A$ (necessarily the identity) such that $b_f(a) = (y,a)$ implies 
  $b_g(a) = (y, h_a(a)) = (y,a)$. But since $y = f(a)$ by definition of $b_f$, and since the above holds for \emph{all} $a \in A$, it thus follows that $f(x) = g(x)$ for all $x \in A$, which in turn implies $f=g$ in \Pfn{} by extensionality. So $F$ is faithful.

  It is easy to verify that $F$ is monoidal.
\end{proof}

\begin{corollary}
  \CPTP{} is the restriction monoidal completion of \Isometry{} quotiented by
  well-pointedness, and \Pfn{} is the restriction monoidal completion of \PInj{}
  quotiented by well-pointedness.
\end{corollary}
\begin{proof}
  Combine Theorems~\ref{thm:lbc-free} and~\ref{thm:well-pointed-completion} with Propositions~\ref{prop:cptp} and~\ref{prop:bennett}.
\end{proof}
% section classical_and_quantum_completions (end)

\section{Cofree reversible foundations} % (fold)
\label{sec:cofree-reversible-foundations}

While \CPTP{} and \Pfn{} both arise as completions of `reversible' categories \Isometry{} and \PInj, it is difficult to pinpoint the  features which make them reversible. For example, \PInj{} is an inverse category, but \Isometry{} is not even a dagger category. Following~\cite{huotstaton:universal}, we peel off another layer to reveal the inverse category underneath using the $\Rb$-construction, the dual to $\Lb$. Thus we can show that both \CPTP{} and \Pfn{} arise via the same universal constructions on the inverse categories \Unitary{} and \PInj{}. We go on to show that this amalgamation of constructions is itself invertible by universal means, allowing us to reconstruct \PInj{} and \Unitary{} from \Pfn{} and \CPTP{} as their cofree (monoidal) inverse categories.

%The $\Rb$-construction is fortunately easy to define and prove properties about, thanks to duality.

\begin{definition}
  For a symmetric monoidal inverse category $\C$, define $\Rb(\C) =
  \Lb(\C^\opp)^\opp$.
\end{definition}

\begin{proposition}
  When $\C$ is a symmetric monoidal inverse category, $\Rb(\C)$ is a
  coaffine symmetric monoidal restriction category.
\end{proposition}
\begin{proof}
  Inverse categories are self-dual, $\C \simeq \C^\opp$, so $\Rb(\C) =
  \Lb(\C^\opp)^\opp \simeq \Lb(\C)^\opp$. Hence $\Lb(\C)$ is an
  affine symmetric monoidal restriction category, and $\Rb(\C)$ is a coaffine symmetric monoidal \emph{corestriction} category. It is also a symmetric monoidal restriction category under $\ridm{[f,E]^\opp} = [\rho^{-1} \circ \ridm{f^\dagger},I]$, because in an inverse category $\C$ morphisms $f$ have (monoidal) corestriction $\ridm{f^\dagger}$.
\end{proof}

The $\Lb$-construction (and, by duality, the $\Rb$-construction) is \emph{conservative}: if a monoidal category is already affine, the construction does nothing (up to isomorphism).

\begin{proposition}\label{prop:conservative}
  If $\C$ is a restriction affine symmetric monoidal category, there is a monoidal equivalence $\Lb(\C) \simeq \C$.
\end{proposition}
\begin{proof}
  It suffices to show that each morphism is equivalent to one of the form $\mathcal{E}(f')$. Let $[f,E] \colon A \to B$ be a morphism of $\Lb(\C)$.
  Then 
  $\ridm{\mathcal{E}(\pi_1 \circ f)} = \mathcal{E}(\ridm{\pi_1 \circ
  f}) = \mathcal{E}(\ridm{f}) = {[\rho^{-1} \circ \ridm{f}, I]} =
  \ridm{[f,E]}$ and:
  % https://q.uiver.app/?q=WzAsNyxbMCwyLCJBIl0sWzEsMSwiQiBcXG90aW1lcyBFIl0sWzIsMSwiQiJdLFszLDEsIkIgXFxvdGltZXMgSSJdLFszLDMsIkIgXFxvdGltZXMgRSJdLFszLDIsIkIgXFxvdGltZXMgSSJdLFsyLDAsIkIgXFxvdGltZXMgSSJdLFswLDEsImYiXSxbMSwyLCJcXHBpXzEiXSxbMiwzLCJcXHJob157LTF9Il0sWzAsNCwiZiIsMix7ImN1cnZlIjoyfV0sWzQsNSwiXFxpZCBcXG90aW1lcyAhIiwyXSxbMyw1LCJcXGlkIFxcb3RpbWVzIFxcaWQiXSxbMCwzLCJcXG1hdGhjYWx7RX0oXFxwaV8xIFxcY2lyYyBmKSIsMix7ImNvbG91ciI6WzAsMCw1MF0sInN0eWxlIjp7ImJvZHkiOnsibmFtZSI6InNxdWlnZ2x5In19fSxbMCwwLDUwLDFdXSxbMCw0LCIoZixFKSIsMCx7ImNvbG91ciI6WzAsMCw1MF0sInN0eWxlIjp7ImJvZHkiOnsibmFtZSI6InNxdWlnZ2x5In19fSxbMCwwLDUwLDFdXSxbMSw2LCJcXGlkIFxcb3RpbWVzICEiXSxbNiwyLCJcXHJobyJdLFs2LDMsIlxcaWQiXV0=
  \[\begin{tikzcd}
  	&& {B \otimes I} \\
  	& {B \otimes E} & B & {B \otimes I} \\
  	A &&& {B \otimes I} \\
  	&&& {B \otimes E}
  	\arrow["f", from=3-1, to=2-2]
  	\arrow["{\pi_1}", from=2-2, to=2-3]
  	\arrow["{\rho^{-1}}", from=2-3, to=2-4]
  	\arrow["f"', curve={height=12pt}, from=3-1, to=4-4]
  	\arrow["{\id \otimes !}"', from=4-4, to=3-4]
  	\arrow["{\id \otimes \id}", from=2-4, to=3-4]
  	\arrow["{\mathcal{E}(\pi_1 \circ f)}"', color={rgb,255:red,128;green,128;blue,128}, squiggly, from=3-1, to=2-4]
  	\arrow["{[f,E]}", color={rgb,255:red,128;green,128;blue,128}, squiggly, from=3-1, to=4-4]
  	\arrow["{\id \otimes !}", from=2-2, to=1-3]
  	\arrow["\rho", from=1-3, to=2-3]
  	\arrow["\id", from=1-3, to=2-4]
  \end{tikzcd}\]
  So $\mathcal{E}(\pi_1 \circ f) \sim [f,E]$.
\end{proof}

We can now show that \Pfn{} and \CPTP{} arise as completions of
the inverse categories \PInj{} and \Unitary{}. The quantum case relies on
Huot and Staton's characterisation of \Isometry{} as a completion of
\Unitary{}~\cite{huotstaton:completion} making initial the \emph{unit of the direct sum}. We consider \PInj{} and \Unitary{} as inverse rig categories, using the $\Rb$-construction to make the unit of
the direct sum initial, and then the $\Lb$-construction to make the tensor unit terminal. In this bimonoidal setting, we will use subscripts to clarify which monoidal structure a construction acts on.

\begin{theorem}
  There are equivalences $\Ext(\LRb{\PInj}) \simeq \Pfn$ and $\Ext(\LRb{\Unitary}) \simeq \CPTP$ of categories.
\end{theorem}
\begin{proof}
  First, that $\Ext(\LRb{\Unitary} \simeq \Ext(L_\otimes(R_\oplus(\Unitary))) \simeq \CPTP$ follows from the fact that $R_\oplus(\Unitary) \simeq \Isometry$ by~\cite[III.3]{huotstaton:completion} and Proposition~\ref{prop:cptp}.
  Now \ifarxiv\linebreak\fi$\Ext(\LRb{\PInj}) \simeq \Pfn$ follows from the unit $0$ of the disjoint sum $\oplus$ in $\PInj$ already being (restriction) 
  initial, so $\Rb_\oplus(\PInj) \simeq \PInj$ by dualising Proposition~\ref{prop:conservative} and finally \ifarxiv\linebreak\fi$\Ext(\LRb{\PInj}) \simeq \Ext(\Lb_\otimes(\PInj)) \simeq \Pfn$.
\end{proof}

Finally, we show that, at least in these two cases, this construction can be undone by considering their \emph{cofree inverse categories} (see Proposition~\ref{prop:cofreeinverse}).
Write $\Unitary_p$ for the category of finite-dimensional Hilbert spaces and equivalence classes of unitary linear maps up to global phase: unitaries $f,g \colon H \to K$ are identified when $f=z \cdot g$ for some $z \in U(1)$~\cite[2.1.4]{heunen:thesis}.

\begin{theorem}\label{thm:cofree}
  There are monoidal equivalences $\Inv(\Pfn) \simeq \PInj$ and $\Inv(\CPTP) \simeq \Unitary_p$.
\end{theorem}
\begin{proof}
  That $\Inv(\Pfn) \simeq \PInj$ is well known; see for example~\cite{cockettlack:restcat1}. With $\CPTP$ a trivial restriction category, we 
  show that $\Unitary_p$ is its cofree groupoid. It suffices to
  show that isomorphisms in \CPTP{} just conjugate with a unitary.
  
  Let $\Lambda \colon \mathcal{B}(H) \to \mathcal{B}(K)$ be an isomorphism in \CPTP, that is, a bijective CPTP map with a CPTP inverse. 
  Notice first that since $\Lambda$ is bijective and $H$ and $K$ finite-dimensional, they must in fact have equal dimension. 
  Second, notice that $\Lambda$ must then preserve pure states, since if 
  $\Lambda(\dyad{\phi}{\phi})$ is some mixed state $\sum_i \alpha_i \rho_i$ then $\dyad{\phi}{\phi} = \Lambda^{-1}(\Lambda(\dyad{\phi}{\phi})) = 
  \Lambda^{-1}(\sum_i \alpha_i \rho_i) = \sum_i \alpha_i \Lambda^{-1}(\rho_i)$,
  contradicting purity of $\dyad{\phi}{\phi}$. But since 
  $\id \otimes \Lambda$ is then also an isomorphism, it too preserves pure 
  states, and so the Choi-state $(\id \otimes \Lambda)(\dyad{\Phi}{\Phi})$ for $\Lambda$ is pure, too. Recall that a Stinespring dilation of a CPTP map can be obtained by purifying its Choi-state, sending the result back through the Choi-Jamiolkowski isomorphism, and tracing out the auxiliary system~\cite{renner:lecturenotes}. Since the Choi-state $(\id \otimes \Lambda)(\dyad{\Phi}{\Phi})$ is already pure, $\Lambda$ must then already be conjugation by some isometry $V$, which must in fact be unitary by surjectivity of $\Lambda$.
\end{proof}

\noindent
\textbf{Acknowledgements}\quad We thank Frederik vom Ende for his clarifying comments on Theorem~\ref{thm:cofree}, Cole Comfort for pointing out related work, and Mathieu Huot for useful feedback.

\bibliographystyle{eptcs}
\bibliography{refs}

\appendix
\section{Deferred proofs}\label{sec:proofs}

\begin{repproposition}{prop:LbCcat}
  $\Lb(\C)$ is a category.
\end{repproposition}
\begin{proof}
  We need to show that composition is associative, unital, and well-defined.
  Let $[f,E] \colon A \to B$, $[g,E'] \colon B \to C$, and $[h,E''] \colon C \to D$ be morphisms of $\Lb(\C)$. That $[h, E''] \circ([g,E'] \circ [f,E])$ is equivalent to $([h, E''] \circ [g,E']) \circ [f,E]$ follows from
  \begin{align*}
    \ridm{\alpha \circ (h \otimes \id) \circ \alpha \circ g \otimes \id 
    \circ f} & =
    \ridm{\alpha \circ \alpha \circ ((h \otimes \id) \otimes \id) \circ g
    \otimes \id \circ f} \\
    & = \ridm{((h \otimes \id) \otimes \id) \circ g \otimes \id \circ f} \\
    & = \ridm{\alpha \circ (\alpha \otimes \id) \circ ((h \otimes \id) \otimes
    \id) \circ g \otimes \id \circ f}
  \end{align*}
  in \C{} and commutativity of the following diagram in $\C$:
  %, using squiggly arrows to indicate the corresponding morphisms in $\Lb(\C)$.
  % https://q.uiver.app/?q=WzAsMTIsWzEsNCwiQiBcXG90aW1lcyBFIl0sWzIsNCwiKEMgXFxvdGltZXMgRScpIFxcb3RpbWVzIEUiXSxbMyw0LCIoKEQgXFxvdGltZXMgRScnKSBcXG90aW1lcyBFJykgXFxvdGltZXMgRSJdLFs0LDQsIihEIFxcb3RpbWVzIChFJycgXFxvdGltZXMgRScpKSBcXG90aW1lcyBFIl0sWzQsMywiRCBcXG90aW1lcyAoKEUnJyBcXG90aW1lcyBFJykgXFxvdGltZXMgRSkiXSxbMCwyLCJBIl0sWzEsMCwiQiBcXG90aW1lcyBFIl0sWzIsMCwiKEMgXFxvdGltZXMgRScpIFxcb3RpbWVzIEUiXSxbMywwLCJDIFxcb3RpbWVzIChFJyBcXG90aW1lcyBFKSJdLFs0LDAsIihEIFxcb3RpbWVzIEUnJykgXFxvdGltZXMgKEUnIFxcb3RpbWVzIEUpIl0sWzQsMSwiRCBcXG90aW1lcyAoRScnIFxcb3RpbWVzIChFJyBcXG90aW1lcyBFKSkiXSxbNCwyLCJEIFxcb3RpbWVzIChFJycgXFxvdGltZXMgKEUnIFxcb3RpbWVzIEUpKSJdLFswLDEsImcgXFxvdGltZXMgXFxpZCIsMl0sWzEsMiwiKGggXFxvdGltZXMgXFxpZCkgXFxvdGltZXMgXFxpZCIsMl0sWzIsMywiXFxhbHBoYSBcXG90aW1lcyBcXGlkIiwyXSxbMyw0LCJcXGFscGhhIiwyXSxbNSwwLCJmIiwyXSxbNSw2LCJmIl0sWzYsNywiZyBcXG90aW1lcyBcXGlkIl0sWzcsOCwiXFxhbHBoYSJdLFs4LDksImggXFxvdGltZXMgXFxpZCJdLFs5LDEwLCJcXGFscGhhIl0sWzUsMTAsIihoLEUnJykgXFxjaXJjICgoZyxFJykgXFxjaXJjIChmLEUpKSIsMCx7ImN1cnZlIjotMiwiY29sb3VyIjpbMCwwLDUwXSwic3R5bGUiOnsiYm9keSI6eyJuYW1lIjoic3F1aWdnbHkifX19LFswLDAsNTAsMV1dLFs1LDQsIigoaCxFJycpIFxcY2lyYyAoZyxFJykpIFxcY2lyYyAoZixFKSIsMix7ImN1cnZlIjoyLCJjb2xvdXIiOlswLDAsNTBdLCJzdHlsZSI6eyJib2R5Ijp7Im5hbWUiOiJzcXVpZ2dseSJ9fX0sWzAsMCw1MCwxXV0sWzEwLDExLCJcXGlkIFxcb3RpbWVzIFxcaWQiXSxbNCwxMSwiXFxpZCBcXG90aW1lcyBcXGFscGhhIiwyXV0=
  \[\begin{tikzcd}
    & {B \otimes E} & {(C \otimes E') \otimes E} & {C \otimes (E' \otimes E)} & {(D \otimes E'') \otimes (E' \otimes E)} \\
    &&&& {D \otimes (E'' \otimes (E' \otimes E))} \\
    A &&&& {D \otimes (E'' \otimes (E' \otimes E))} \\
    &&&& {D \otimes ((E'' \otimes E') \otimes E)} \\
    & {B \otimes E} & {(C \otimes E') \otimes E} & {((D \otimes E'') \otimes E') \otimes E} & {(D \otimes (E'' \otimes E')) \otimes E}
    \arrow["{g \otimes \id}"', from=5-2, to=5-3]
    \arrow["{(h \otimes \id) \otimes \id}"', from=5-3, to=5-4]
    \arrow["{\alpha \otimes \id}"', from=5-4, to=5-5]
    \arrow["\alpha"', from=5-5, to=4-5]
    \arrow["f"', from=3-1, to=5-2]
    \arrow["f", from=3-1, to=1-2]
    \arrow["{g \otimes \id}", from=1-2, to=1-3]
    \arrow["\alpha", from=1-3, to=1-4]
    \arrow["{h \otimes \id}", from=1-4, to=1-5]
    \arrow["\alpha", from=1-5, to=2-5]
    \arrow["{[h,E''] \circ ([g,E'] \circ [f,E])}", color={rgb,255:red,128;green,128;blue,128}, curve={height=-12pt}, squiggly, from=3-1, to=2-5]
    \arrow["{([h,E''] \circ [g,E']) \circ [f,E]}"', color={rgb,255:red,128;green,128;blue,128}, curve={height=12pt}, squiggly, from=3-1, to=4-5]
    \arrow["{\id \otimes \id}", from=2-5, to=3-5]
    \arrow["{\id \otimes \alpha}"', from=4-5, to=3-5]
  \end{tikzcd}\]

  That $\id \circ [f,E] \sim [f,E]$ follows from $\ridm{\alpha \circ (\rho^{-1} \otimes \id) \circ f} = \ridm{f}$ and commutativity in \C{} of the diagram:
  % https://q.uiver.app/?q=WzAsNixbMCwxLCJBIl0sWzEsMCwiQiBcXG90aW1lcyBFIl0sWzIsMCwiKEIgXFxvdGltZXMgSSkgXFxvdGltZXMgRSJdLFszLDAsIkIgXFxvdGltZXMgKEkgXFxvdGltZXMgRSkiXSxbMywyLCJCIFxcb3RpbWVzIEUiXSxbMywxLCJCIFxcb3RpbWVzIEUiXSxbMCw0LCJmIiwwLHsiY3VydmUiOjN9XSxbMCwxLCJmIl0sWzEsMiwiXFxyaG9eey0xfSBcXG90aW1lcyBcXGlkIl0sWzIsMywiXFxhbHBoYSJdLFswLDMsIlxcaWRcXGNpcmMgKGYsRSkiLDIseyJjb2xvdXIiOlswLDAsNTBdLCJzdHlsZSI6eyJib2R5Ijp7Im5hbWUiOiJzcXVpZ2dseSJ9fX0sWzAsMCw1MCwxXV0sWzMsNSwiXFxpZCBcXG90aW1lcyBcXGxhbWJkYSJdLFs0LDUsIlxcaWQgXFxvdGltZXMgXFxpZCIsMl0sWzAsNCwiKGYsRSkiLDAseyJjb2xvdXIiOlswLDAsNTBdLCJzdHlsZSI6eyJib2R5Ijp7Im5hbWUiOiJzcXVpZ2dseSJ9fX0sWzAsMCw1MCwxXV1d
  \[\begin{tikzcd}
    & {B \otimes E} & {(B \otimes I) \otimes E} & {B \otimes (I \otimes E)} \\
    A &&& {B \otimes E} \\
    &&& {B \otimes E}
    \arrow["f", curve={height=18pt}, from=2-1, to=3-4]
    \arrow["f", from=2-1, to=1-2]
    \arrow["{\rho^{-1} \otimes \id}", from=1-2, to=1-3]
    \arrow["\alpha", from=1-3, to=1-4]
    \arrow["{\id\circ [f,E]}"', color={rgb,255:red,128;green,128;blue,128}, squiggly, from=2-1, to=1-4]
    \arrow["{\id \otimes \lambda}", from=1-4, to=2-4]
    \arrow["{\id \otimes \id}"', from=3-4, to=2-4]
    \arrow["{[f,E]}", color={rgb,255:red,128;green,128;blue,128}, squiggly, from=2-1, to=3-4]
  \end{tikzcd}\]
  Similarly $[f,E] \circ \id \sim [f,E]$.
  Finally, we show that composition is well-defined. 
  Suppose $[f,E] \colon A \to B$ is equivalent to $[f',G] \colon A \to B$ by a zigzag of mediators $E \tot{h_1} E_1 \fromt{h_2} E_2 \tot{h_3} \cdots \fromt{h_n} G$.
  Given $[g,E'] \colon B \to C$ and intermediates $f_1, f_2, \dots, f_{n-1}$, to show $[g,E'] \circ [f,E] \sim [g,E'] \circ [f',G]$ we see first
  that the diagram below commutes in \C:
  % https://q.uiver.app/?q=WzAsMTYsWzIsMCwiQSJdLFswLDEsIkIgXFxvdGltZXMgRSJdLFs1LDEsIkIgXFxvdGltZXMgRyJdLFswLDIsIihDIFxcb3RpbWVzIEUnKSBcXG90aW1lcyBFIl0sWzUsMiwiKEMgXFxvdGltZXMgRScpIFxcb3RpbWVzIEciXSxbMCwzLCJDIFxcb3RpbWVzIChFJyBcXG90aW1lcyBFKSJdLFs1LDMsIkMgXFxvdGltZXMgKEUnIFxcb3RpbWVzIEcpIl0sWzEsMSwiQiBcXG90aW1lcyBFXzEiXSxbMiwxLCJCIFxcb3RpbWVzIEVfMiJdLFszLDEsIlxcZG90cyJdLFsxLDMsIkMgXFxvdGltZXMgKEUnIFxcb3RpbWVzIEVfMSkiXSxbMiwzLCJDIFxcb3RpbWVzIChFJyBcXG90aW1lcyBFXzIpIl0sWzMsMywiXFxkb3RzIl0sWzEsMiwiKEMgXFxvdGltZXMgRScpIFxcb3RpbWVzIEVfMSJdLFsyLDIsIihDIFxcb3RpbWVzIEUnKSBcXG90aW1lcyBFXzIiXSxbMywyLCJcXGRvdHMiXSxbMCwxLCJmIiwyXSxbMCwyLCJmJyJdLFszLDUsIlxcYWxwaGEiLDJdLFsxLDMsImcgXFxvdGltZXMgXFxpZCIsMl0sWzIsNCwiZyBcXG90aW1lcyBcXGlkIl0sWzQsNiwiXFxhbHBoYSJdLFsxLDcsIlxcaWQgXFxvdGltZXMgaF8xIiwyXSxbOCw3LCJcXGlkIFxcb3RpbWVzIGhfMiJdLFsyLDksIlxcaWQgXFxvdGltZXMgaF9uIl0sWzgsOSwiXFxpZCBcXG90aW1lcyBoXzMiLDJdLFs1LDEwLCJcXGlkIFxcb3RpbWVzIChcXGlkIFxcb3RpbWVzIGhfMSkiLDJdLFsxMSwxMCwiXFxpZCBcXG90aW1lcyAoXFxpZCBcXG90aW1lcyBoXzIpIl0sWzExLDEyLCJcXGlkIFxcb3RpbWVzIChcXGlkIFxcb3RpbWVzIGhfMykiLDJdLFs2LDEyLCJcXGlkIFxcb3RpbWVzIChcXGlkIFxcb3RpbWVzIGhfbikiXSxbNywxMywiZyBcXG90aW1lcyBcXGlkIl0sWzgsMTQsImcgXFxvdGltZXMgXFxpZCJdLFsxMywxMCwiXFxhbHBoYSJdLFsxNCwxMSwiXFxhbHBoYSJdLFszLDEzLCJcXGlkIFxcb3RpbWVzIGhfMSIsMl0sWzQsMTUsIlxcaWQgXFxvdGltZXMgaF9uIl0sWzE0LDE1LCJcXGlkIFxcb3RpbWVzIGhfMyIsMl0sWzE0LDEzLCJcXGlkIFxcb3RpbWVzIGhfMiJdLFswLDgsImZfMiBcXHF1YWQgXFxjZG90cyJdLFswLDcsImZfMSJdXQ==
  \[\begin{tikzcd}
    && A \\
    {B \otimes E} & {B \otimes E_1} & {B \otimes E_2} & \dots && {B \otimes G} \\
    {(C \otimes E') \otimes E} & {(C \otimes E') \otimes E_1} & {(C \otimes E') \otimes E_2} & \dots && {(C \otimes E') \otimes G} \\
    {C \otimes (E' \otimes E)} & {C \otimes (E' \otimes E_1)} & {C \otimes (E' \otimes E_2)} & \dots && {C \otimes (E' \otimes G)}
    \arrow["f"', from=1-3, to=2-1]
    \arrow["{f'}", from=1-3, to=2-6]
    \arrow["\alpha"', from=3-1, to=4-1]
    \arrow["{g \otimes \id}"', from=2-1, to=3-1]
    \arrow["{g \otimes \id}", from=2-6, to=3-6]
    \arrow["\alpha", from=3-6, to=4-6]
    \arrow["{\id \otimes h_1}"', from=2-1, to=2-2]
    \arrow["{\id \otimes h_2}", from=2-3, to=2-2]
    \arrow["{\id \otimes h_n}", from=2-6, to=2-4]
    \arrow["{\id \otimes h_3}"', from=2-3, to=2-4]
    \arrow["{\id \otimes (\id \otimes h_1)}"', from=4-1, to=4-2]
    \arrow["{\id \otimes (\id \otimes h_2)}", from=4-3, to=4-2]
    \arrow["{\id \otimes (\id \otimes h_3)}"', from=4-3, to=4-4]
    \arrow["{\id \otimes (\id \otimes h_n)}", from=4-6, to=4-4]
    \arrow["{g \otimes \id}", from=2-2, to=3-2]
    \arrow["{g \otimes \id}", from=2-3, to=3-3]
    \arrow["\alpha", from=3-2, to=4-2]
    \arrow["\alpha", from=3-3, to=4-3]
    \arrow["{\id \otimes h_1}"', from=3-1, to=3-2]
    \arrow["{\id \otimes h_n}", from=3-6, to=3-4]
    \arrow["{\id \otimes h_3}"', from=3-3, to=3-4]
    \arrow["{\id \otimes h_2}", from=3-3, to=3-2]
    \arrow["{f_2 \quad \cdots}", from=1-3, to=2-3]
    \arrow["{f_1}", from=1-3, to=2-2]
  \end{tikzcd}\]
  There is no room in the diagram above for ghost arrows, but each downward 
  path $\alpha \circ (g \otimes \id) \circ f_i$
  corresponds to $[g, E'] \circ [f_i, E_i]$, and likewise for $[f,E]$ and
  $[f',G]$ instead of $f_i$.
  We have left to show that $\ridm{\alpha \circ (g \otimes \id) \circ f} = 
  \ridm{\alpha \circ (g \otimes \id) \circ f'}$. Now  $\ridm{\id \otimes h_1} \circ f = f \circ \ridm{(\id \otimes h_1) \circ f} = f \circ \ridm{f_1} = f \circ \ridm{f} = f$, so:
  \begin{align*}
    \ridm{\alpha \circ (g \otimes \id) \circ f} & =
    \ridm{(g \otimes \id) \circ f} = \ridm{(g \otimes \id) \circ
    \ridm{\id \otimes h_1} \circ f} = \ridm{(g \otimes \id) \circ
    (\id \otimes \ridm{h_1}) \circ f} \\ & =
    \ridm{(g \otimes \ridm{h_1}) \circ f} =
    \ridm{\ridm{g \otimes h_1} \circ f} =
    \ridm{(g \otimes h_1) \circ f} =
    \ridm{(g \otimes \id) \circ (\id \otimes h_1) \circ f} \\
    & = \ridm{(g \otimes \id) \circ f_1}
    = \ridm{\alpha \circ (g \otimes \id) \circ f_1}
  \end{align*}
  By induction eventually $\ridm{\alpha \circ (g \otimes \id) \circ f} = \ridm{\alpha \circ (g \otimes \id) \circ f'}$.
  
  Pre-composition is similarly well-defined, though the condition on restriction idempotents follows more readily by 
  $\ridm{(g \otimes \id) \circ f} = \ridm{\ridm{(g \otimes \id)} \circ f}
  = \ridm{\ridm{(g' \otimes \id)} \circ f} = \ridm{(g' \otimes \id) \circ f}$.
\end{proof}

\begin{repproposition}{prop:LbCrestriction}
  $\Lb(\C)$ inherits a restriction structure from $\C$ with $\ridm{[f,E]} 
  = [\rho^{-1} \circ \ridm{f}, I]$.
\end{repproposition}
\begin{proof}
  We establish the axioms of Definition~\ref{def:restrictioncategory} in order. 
  That $[f,E] \circ \ridm{[f,E]} = [f,E]$ for each $[f,E] \colon A \to B$ in $\Lb(\C)$ follows by commutativity of the following diagram in \C:
  % https://q.uiver.app/?q=WzAsOCxbMCwyLCJBIl0sWzUsMywiQiBcXG90aW1lcyBFIl0sWzIsMSwiQSJdLFszLDEsIkEgXFxvdGltZXMgSSJdLFs0LDEsIihCIFxcb3RpbWVzIEUpIFxcb3RpbWVzIEkiXSxbNSwxLCJCIFxcb3RpbWVzIChFIFxcb3RpbWVzIEkpIl0sWzUsMiwiQiBcXG90aW1lcyBFIl0sWzIsMCwiQiBcXG90aW1lcyBFIl0sWzAsMSwiZiIsMix7ImN1cnZlIjozfV0sWzAsMiwiXFxyaWRte2Z9Il0sWzIsMywiXFxyaG9eey0xfSJdLFszLDQsImYgXFxvdGltZXMgXFxpZCJdLFs0LDUsIlxcYWxwaGEiXSxbNSw2LCJcXGlkIFxcb3RpbWVzIFxccmhvIl0sWzEsNiwiXFxpZCBcXG90aW1lcyBcXGlkIiwyXSxbMCw1LCIoZixFKSBcXGNpcmMgXFxyaWRteyhmLEUpfSIsMix7ImNvbG91ciI6WzAsMCw1MF0sInN0eWxlIjp7ImJvZHkiOnsibmFtZSI6InNxdWlnZ2x5In19fSxbMCwwLDUwLDFdXSxbMCwxLCIoZixFKSIsMCx7ImNvbG91ciI6WzAsMCw1MF0sInN0eWxlIjp7ImJvZHkiOnsibmFtZSI6InNxdWlnZ2x5In19fSxbMCwwLDUwLDFdXSxbMiw3LCJmIl0sWzcsNCwiXFxyaG9eey0xfSIsMCx7ImN1cnZlIjotMn1dLFswLDcsImYiLDAseyJjdXJ2ZSI6LTJ9XSxbNCw2LCJcXHJobyIsMl1d
  \[\begin{tikzcd}
    && {B \otimes E} \\
    && A & {A \otimes I} & {(B \otimes E) \otimes I} & {B \otimes (E \otimes I)} \\
    A &&&&& {B \otimes E} \\
    &&&&& {B \otimes E}
    \arrow["f"', curve={height=18pt}, from=3-1, to=4-6]
    \arrow["{\ridm{f}}", from=3-1, to=2-3]
    \arrow["{\rho^{-1}}", from=2-3, to=2-4]
    \arrow["{f \otimes \id}", from=2-4, to=2-5]
    \arrow["\alpha", from=2-5, to=2-6]
    \arrow["{\id \otimes \rho}", from=2-6, to=3-6]
    \arrow["{\id \otimes \id}"', from=4-6, to=3-6]
    \arrow["{[f,E] \circ \ridm{[f,E]}}"', color={rgb,255:red,128;green,128;blue,128}, squiggly, from=3-1, to=2-6]
    \arrow["{[f,E]}", color={rgb,255:red,128;green,128;blue,128}, squiggly, from=3-1, to=4-6]
    \arrow["f", from=2-3, to=1-3]
    \arrow["{\rho^{-1}}", curve={height=-12pt}, from=1-3, to=2-5]
    \arrow["f", curve={height=-12pt}, from=3-1, to=1-3]
    \arrow["\rho"', from=2-5, to=3-6]
  \end{tikzcd}\]
  To see that $\ridm{[f,E]} \circ \ridm{[g,E']} = \ridm{[g,E']} \circ \ridm{[f,E]}$ for $[f,E] \colon A \to B$ and $[g,E'] \colon A \to C$ in $\Lb(\C)$:
  % https://q.uiver.app/?q=WzAsMTYsWzAsMiwiQSJdLFsyLDAsIkEiXSxbMiw0LCJBIl0sWzMsNCwiQSBcXG90aW1lcyBJIl0sWzMsMywiQSJdLFszLDEsIkEiXSxbNCw0LCJBIFxcb3RpbWVzIEkiXSxbMywwLCJBIFxcb3RpbWVzIEkiXSxbNCwwLCJBIFxcb3RpbWVzIEkiXSxbNSwwLCIoQSBcXG90aW1lcyBJKSBcXG90aW1lcyBJIl0sWzUsNCwiKEEgXFxvdGltZXMgSSkgXFxvdGltZXMgSSJdLFs1LDMsIkHCoFxcb3RpbWVzIChJIFxcb3RpbWVzIEkpIl0sWzUsMSwiQSBcXG90aW1lcyAoSSBcXG90aW1lcyBJKSJdLFs1LDIsIkEgXFxvdGltZXMgKEkgXFxvdGltZXMgSSkiXSxbMiwzLCJBIl0sWzIsMSwiQSJdLFswLDEsIlxccmlkbXtmfSIsMCx7ImN1cnZlIjotM31dLFswLDIsIlxccmlkbXtnfSIsMix7ImN1cnZlIjozfV0sWzIsMywiXFxyaG9eey0xfSIsMl0sWzIsNCwiXFxyaWRte2Z9Il0sWzEsNSwiXFxyaWRte2d9IiwyXSxbNCw2LCJcXHJob157LTF9Il0sWzMsNiwiXFxyaWRte2Z9wqBcXG90aW1lcyBcXGlkIiwyXSxbMSw3LCJcXHJob157LTF9Il0sWzcsOCwiXFxyaWRte2d9IFxcb3RpbWVzIFxcaWQiXSxbNSw4LCJcXHJob157LTF9IiwyXSxbOCw5LCJcXHJob157LTF9wqBcXG90aW1lcyBcXGlkIl0sWzYsMTAsIlxccmhvXnstMX3CoFxcb3RpbWVzIFxcaWQiLDJdLFsxMCwxMSwiXFxhbHBoYSIsMl0sWzksMTIsIlxcYWxwaGEiXSxbMCwxMiwiXFxyaWRteyhnLEUnKX0gXFxjaXJjIFxccmlkbXsoZixFKX0iLDIseyJjb2xvdXIiOlswLDAsNTBdLCJzdHlsZSI6eyJib2R5Ijp7Im5hbWUiOiJzcXVpZ2dseSJ9fX0sWzAsMCw1MCwxXV0sWzAsMTEsIlxccmlkbXsoZixFKX0gXFxjaXJjIFxccmlkbXsoZyxFJyl9IiwwLHsiY29sb3VyIjpbMCwwLDUwXSwic3R5bGUiOnsiYm9keSI6eyJuYW1lIjoic3F1aWdnbHkifX19LFswLDAsNTAsMV1dLFsxMSwxMywiXFxpZCBcXG90aW1lcyBcXGlkIiwyXSxbMTIsMTMsIlxcaWQgXFxvdGltZXMgXFxpZCJdLFswLDE0LCJcXHJpZG17Zn0iLDJdLFsxNCw0LCJcXHJpZG17Z30iXSxbMCwxNSwiXFxyaWRte2d9Il0sWzE1LDUsIlxccmlkbXtmfSIsMl1d
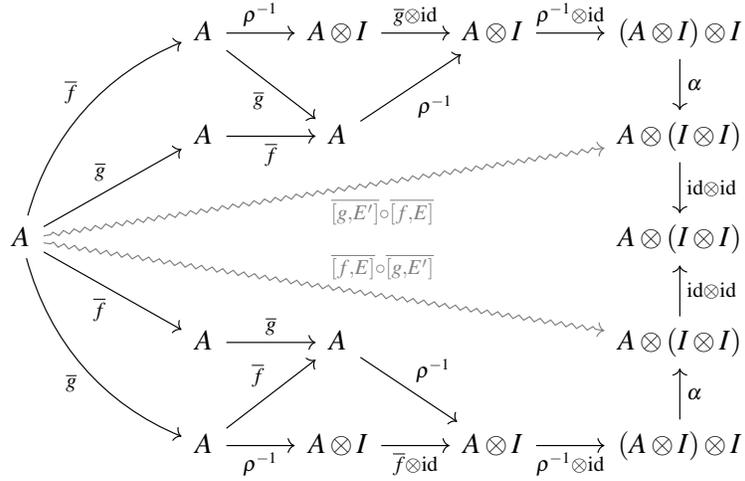
  \[\begin{tikzcd}
    && A & {A \otimes I} & {A \otimes I} & {(A \otimes I) \otimes I} \\
    && A & A && {A \otimes (I \otimes I)} \\
    A &&&&& {A \otimes (I \otimes I)} \\
    && A & A && {A \otimes (I \otimes I)} \\
    && A & {A \otimes I} & {A \otimes I} & {(A \otimes I) \otimes I}
    \arrow["{\ridm{f}}", curve={height=-18pt}, from=3-1, to=1-3]
    \arrow["{\ridm{g}}"', curve={height=18pt}, from=3-1, to=5-3]
    \arrow["{\rho^{-1}}"', from=5-3, to=5-4]
    \arrow["{\ridm{f}}", from=5-3, to=4-4]
    \arrow["{\ridm{g}}"', from=1-3, to=2-4]
    \arrow["{\rho^{-1}}", from=4-4, to=5-5]
    \arrow["{\ridm{f} \otimes \id}"', from=5-4, to=5-5]
    \arrow["{\rho^{-1}}", from=1-3, to=1-4]
    \arrow["{\ridm{g} \otimes \id}", from=1-4, to=1-5]
    \arrow["{\rho^{-1}}"', from=2-4, to=1-5]
    \arrow["{\rho^{-1} \otimes \id}", from=1-5, to=1-6]
    \arrow["{\rho^{-1} \otimes \id}"', from=5-5, to=5-6]
    \arrow["\alpha"', from=5-6, to=4-6]
    \arrow["\alpha", from=1-6, to=2-6]
    \arrow["{\ridm{[g,E']} \circ \ridm{[f,E]}}"', color={rgb,255:red,128;green,128;blue,128}, squiggly, from=3-1, to=2-6]
    \arrow["{\ridm{[f,E]} \circ \ridm{[g,E']}}", color={rgb,255:red,128;green,128;blue,128}, squiggly, from=3-1, to=4-6]
    \arrow["{\id \otimes \id}"', from=4-6, to=3-6]
    \arrow["{\id \otimes \id}", from=2-6, to=3-6]
    \arrow["{\ridm{f}}"', from=3-1, to=4-3]
    \arrow["{\ridm{g}}", from=4-3, to=4-4]
    \arrow["{\ridm{g}}", from=3-1, to=2-3]
    \arrow["{\ridm{f}}"', from=2-3, to=2-4]
  \end{tikzcd}\]
  To show $\ridm{[g,E'] \circ \ridm{[f,E]}} = \ridm{[g,E']} \circ \ridm{[f,E]}$ for all $[f,E] \colon A \to B$ and $[g,E'] \colon A \to C$ of $\Lb(\C)$, first compute:
  \begin{align*}
    \ridm{[g,E'] \circ \ridm{[f,E]}} & = 
    \ridm{[g,E'] \circ (\rho^{-1} \circ \ridm{f},I)} =
    \ridm{(\alpha \circ (g \otimes \id) \circ \rho^{-1} \circ \ridm{f},I)} =
    \ridm{(\alpha \circ \rho^{-1} \circ g \circ \ridm{f},I)} \\
    & = (\rho^{-1} \circ \ridm{\alpha \circ \rho^{-1} \circ g \circ 
    \ridm{f}},I) = (\rho^{-1} \circ \ridm{g \circ \ridm{f}},I) 
  \end{align*}
  Now the diagram below commutes in $\C$ because $\ridm{g \circ \ridm{f}} = 
  \ridm{g} \circ \ridm{f}$:
  % https://q.uiver.app/?q=WzAsMTAsWzAsMiwiQSJdLFsxLDEsIkEiXSxbMSwzLCJBIl0sWzUsMywiQSBcXG90aW1lcyBJIl0sWzIsMSwiQSBcXG90aW1lcyBJIl0sWzMsMSwiQSBcXG90aW1lcyBJIl0sWzQsMSwiKEEgXFxvdGltZXMgSSkgXFxvdGltZXMgSSJdLFs1LDEsIkEgXFxvdGltZXMgKEkgXFxvdGltZXMgSSkiXSxbNSwyLCJBIFxcb3RpbWVzIEkiXSxbMiwwLCJBIl0sWzAsMSwiXFxyaWRte2Z9Il0sWzAsMiwiXFxyaWRte2cgXFxjaXJjIFxccmlkbXtmfX0iLDJdLFsyLDMsIlxccmhvXnstMX0iLDJdLFsxLDQsIlxccmhvXnstMX0iXSxbNCw1LCJcXHJpZG17Z30gXFxvdGltZXPCoFxcaWQiXSxbNSw2LCJcXHJob157LTF9wqBcXG90aW1lcyBcXGlkIl0sWzYsNywiXFxhbHBoYSJdLFs3LDgsIlxcaWQgXFxvdGltZXMgXFxyaG8iXSxbMyw4LCJcXGlkIFxcb3RpbWVzIFxcaWQiLDJdLFsxLDksIlxccmlkbXtnfSJdLFs5LDUsIlxccmhvXnstMX0iXSxbMCw3LCJcXHJpZG17KGcsRScpfSBcXGNpcmMgXFxyaWRteyhmLEUpfSIsMix7ImNvbG91ciI6WzAsMCw1MF0sInN0eWxlIjp7ImJvZHkiOnsibmFtZSI6InNxdWlnZ2x5In19fSxbMCwwLDUwLDFdXSxbMCwzLCJcXHJpZG17KGcsRScpIFxcY2lyYyBcXHJpZG17KGYsRSl9fSIsMCx7ImNvbG91ciI6WzAsMCw1MF0sInN0eWxlIjp7ImJvZHkiOnsibmFtZSI6InNxdWlnZ2x5In19fSxbMCwwLDUwLDFdXV0=
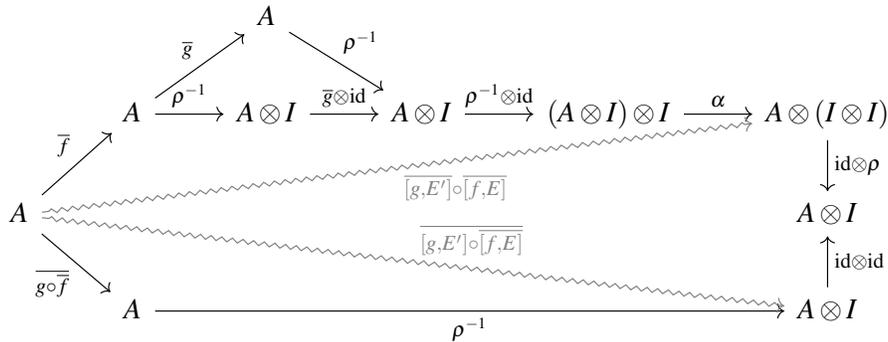
  \[\begin{tikzcd}
    && A \\
    & A & {A \otimes I} & {A \otimes I} & {(A \otimes I) \otimes I} & {A \otimes (I \otimes I)} \\
    A &&&&& {A \otimes I} \\
    & A &&&& {A \otimes I}
    \arrow["{\ridm{f}}", from=3-1, to=2-2]
    \arrow["{\ridm{g \circ \ridm{f}}}"', from=3-1, to=4-2]
    \arrow["{\rho^{-1}}"', from=4-2, to=4-6]
    \arrow["{\rho^{-1}}", from=2-2, to=2-3]
    \arrow["{\ridm{g} \otimes \id}", from=2-3, to=2-4]
    \arrow["{\rho^{-1} \otimes \id}", from=2-4, to=2-5]
    \arrow["\alpha", from=2-5, to=2-6]
    \arrow["{\id \otimes \rho}", from=2-6, to=3-6]
    \arrow["{\id \otimes \id}"', from=4-6, to=3-6]
    \arrow["{\ridm{g}}", from=2-2, to=1-3]
    \arrow["{\rho^{-1}}", from=1-3, to=2-4]
    \arrow["{\ridm{[g,E']} \circ \ridm{[f,E]}}"', color={rgb,255:red,128;green,128;blue,128}, squiggly, from=3-1, to=2-6]
    \arrow["{\ridm{[g,E'] \circ \ridm{[f,E]}}}", color={rgb,255:red,128;green,128;blue,128}, squiggly, from=3-1, to=4-6]
  \end{tikzcd}\]
  Finally, for $[f,E] \colon A \to B$ and $[g,E'] \colon B \to C$ we have 
  $\ridm{[g,E']} \circ [f,E] = [f,E] \circ \ridm{[g,E'] \circ [f,E]}$ because $\ridm{[g,E'] \circ [f,E]} = [\ridm{\alpha \circ (g \otimes \id) \circ f},I] = [\ridm{(g \otimes \id) \circ f},I] = [\ridm{(g \otimes \id) \circ f},I]$ and the diagram below commutes:
  % https://q.uiver.app/?q=WzAsMTIsWzAsMSwiQSJdLFsxLDAsIkIgXFxvdGltZXMgRSJdLFs0LDAsIihCIFxcb3RpbWVzIEkpwqBcXG90aW1lcyBFIl0sWzUsMCwiQiBcXG90aW1lcyAoSSBcXG90aW1lcyBFKSJdLFsxLDIsIkEiXSxbMywzLCJBIFxcb3RpbWVzIEkiXSxbNCwyLCIoQiBcXG90aW1lcyBFKSBcXG90aW1lcyBJIl0sWzUsMiwiQiBcXG90aW1lcyAoRSBcXG90aW1lcyBJKSJdLFsyLDNdLFs1LDEsIkIgXFxvdGltZXMgRSJdLFszLDIsIkIgXFxvdGltZXMgRSJdLFszLDAsIkIgXFxvdGltZXMgRSJdLFswLDEsImYiXSxbMiwzLCJcXGFscGhhIl0sWzAsNCwiXFxyaWRteyhnIFxcb3RpbWVzIFxcaWQpIFxcY2lyYyBmfSIsMl0sWzQsNSwiXFxyaG9eey0xfSIsMl0sWzUsNiwiZiBcXG90aW1lcyBcXGlkIiwyXSxbNiw3LCJcXGFscGhhIiwyXSxbMyw5LCJcXGlkIFxcb3RpbWVzIFxcbGFtYmRhIl0sWzcsOSwiXFxpZCBcXG90aW1lcyBcXHJobyIsMl0sWzQsMTAsImYiXSxbMTAsNiwiXFxyaG9eey0xfSJdLFsxLDExLCJcXHJpZG17Z30gXFxvdGltZXMgXFxpZCJdLFsxMSwyLCJcXHJob157LTF9IFxcb3RpbWVzIFxcaWQiXSxbMCw3LCIoZixFKSBcXGNpcmMgXFxyaWRteyhnLEUnKSBcXGNpcmMgKGYsRSl9IiwwLHsiY29sb3VyIjpbMCwwLDUwXSwic3R5bGUiOnsiYm9keSI6eyJuYW1lIjoic3F1aWdnbHkifX19LFswLDAsNTAsMV1dLFswLDMsIlxccmlkbXsoZyxFJyl9IFxcY2lyYyAoZixFKSIsMix7ImNvbG91ciI6WzAsMCw0OV0sInN0eWxlIjp7ImJvZHkiOnsibmFtZSI6InNxdWlnZ2x5In19fSxbMCwwLDQ5LDFdXV0=
  \[\begin{tikzcd}
    & {B \otimes E} && {B \otimes E} & {(B \otimes I) \otimes E} & {B \otimes (I \otimes E)} \\
    A &&&&& {B \otimes E} \\
    & A && {B \otimes E} & {(B \otimes E) \otimes I} & {B \otimes (E \otimes I)} \\
    && {} & {A \otimes I}
    \arrow["f", from=2-1, to=1-2]
    \arrow["\alpha", from=1-5, to=1-6]
    \arrow["{\ridm{(g \otimes \id) \circ f}}"', from=2-1, to=3-2]
    \arrow["{\rho^{-1}}"', from=3-2, to=4-4]
    \arrow["{f \otimes \id}"', from=4-4, to=3-5]
    \arrow["\alpha"', from=3-5, to=3-6]
    \arrow["{\id \otimes \lambda}", from=1-6, to=2-6]
    \arrow["{\id \otimes \rho}"', from=3-6, to=2-6]
    \arrow["f", from=3-2, to=3-4]
    \arrow["{\rho^{-1}}", from=3-4, to=3-5]
    \arrow["{\ridm{g} \otimes \id}", from=1-2, to=1-4]
    \arrow["{\rho^{-1} \otimes \id}", from=1-4, to=1-5]
    \arrow["{[f,E] \circ \ridm{[g,E'] \circ [f,E]}}", color={rgb,255:red,128;green,128;blue,128}, squiggly, from=2-1, to=3-6]
    \arrow["{\ridm{[g,E']} \circ [f,E]}"', color={rgb,255:red,125;green,125;blue,125}, squiggly, from=2-1, to=1-6]
  \end{tikzcd}\]
  Here $(\ridm{g} \otimes \id) \circ f = \ridm{g \otimes \id} \circ f = f \circ \ridm{(g \otimes \id) \circ f}$ by the corresponding axiom in \C.
\end{proof}

\begin{repproposition}{prop:LbCmonoidal}
  If \C{} is a restriction symmetric monoidal category, then so is $\Lb(\C)$:
  \begin{itemize}
    \item the tensor unit and tensor product of objects are as in $\C$;
    \item the tensor product of $[f,E] \colon A \to B$ and $[f',e'] \colon A' \to B'$ is $[\vartheta \circ (f \otimes f'), E \otimes E'] \colon A \otimes A' \to B \otimes B'$;
  \end{itemize} 
  where $\vartheta$ is the canonical isomorphism $(B \otimes E) \otimes (B' \otimes E') \simeq (B \otimes B') \otimes (E \otimes E')$ in \C.
\end{repproposition}
\begin{proof}
  Coherence isomorphisms $\rho \colon A \to B$ of $\C$ lift to $\Lb(\C)$ as $[\rho^{-1} \circ \beta, I] \colon A \to B$. For example, the symmetry $\gamma \colon A \otimes B \to B \otimes A$ in $\C$ becomes $[\rho^{-1} \circ \gamma,I] \colon A \otimes B \to B \otimes A$ in $\Lb(\C)$. Composing coherence isomorphisms $[\rho^{-1} \circ \beta,I] \colon A \to B$ and $[\rho^{-1} \circ \phi,I] \colon B \to C$ in $\Lb(\C)$ is equivalent to first composing them in $\C$ and then lifting to $\Lb(\C)$:
  % https://q.uiver.app/?q=WzAsMTAsWzAsMSwiQSJdLFsxLDAsIkIiXSxbNCwwLCIoQyBcXG90aW1lcyBJKSBcXG90aW1lcyBJIl0sWzUsMCwiQyBcXG90aW1lcyAoSSBcXG90aW1lcyBJKSJdLFsyLDAsIkIgXFxvdGltZXMgSSJdLFszLDAsIkMgXFxvdGltZXMgSSJdLFsxLDIsIkIiXSxbMiwyLCJDIl0sWzUsMiwiQyBcXG90aW1lcyBJIl0sWzUsMSwiQyBcXG90aW1lcyBJIl0sWzAsMSwiXFxiZXRhIl0sWzEsNCwiXFxyaG9eey0xfSJdLFs0LDUsIlxccGhpIFxcb3RpbWVzIFxcaWQiXSxbNSwyLCJcXHJob157LTF9IFxcb3RpbWVzIFxcaWQiXSxbMiwzLCJcXGFscGhhIl0sWzAsNiwiXFxiZXRhIiwyXSxbNiw3LCJcXHBoaSIsMl0sWzcsOCwiXFxyaG9eey0xfSIsMl0sWzMsOSwiXFxpZCBcXG90aW1lcyBcXHJobyJdLFs4LDksIlxcaWQgXFxvdGltZXMgXFxpZCIsMl0sWzAsOCwiKFxccmhvXnstMX0gXFxjaXJjIFxccGhpIFxcY2lyYyBcXGJldGEsIEkpIiwwLHsiY29sb3VyIjpbMCwwLDUwXSwic3R5bGUiOnsiYm9keSI6eyJuYW1lIjoic3F1aWdnbHkifX19LFswLDAsNTAsMV1dLFswLDMsIihcXHJob157LTF9IFxcY2lyYyBcXHBoaSwgSSkgXFxjaXJjIChcXHJob157LTF9IFxcY2lyYyBcXGJldGEsIEkpIiwyLHsiY29sb3VyIjpbMCwwLDUwXSwic3R5bGUiOnsiYm9keSI6eyJuYW1lIjoic3F1aWdnbHkifX19LFswLDAsNTAsMV1dXQ==
  \[\begin{tikzcd}[row sep=5mm]
    & B & {B \otimes I} & {C \otimes I} & {(C \otimes I) \otimes I} & {C \otimes (I \otimes I)} \\
    A &&&&& {C \otimes I} \\
    & B & C &&& {C \otimes I}
    \arrow["\beta", from=2-1, to=1-2]
    \arrow["{\rho^{-1}}", from=1-2, to=1-3]
    \arrow["{\phi \otimes \id}", from=1-3, to=1-4]
    \arrow["{\rho^{-1} \otimes \id}", from=1-4, to=1-5]
    \arrow["\alpha", from=1-5, to=1-6]
    \arrow["\beta"', from=2-1, to=3-2]
    \arrow["\phi"', from=3-2, to=3-3]
    \arrow["{\rho^{-1}}"', from=3-3, to=3-6]
    \arrow["{\id \otimes \rho}", from=1-6, to=2-6]
    \arrow["{\id \otimes \id}"', from=3-6, to=2-6]
    \arrow["{[\rho^{-1} \circ \phi \circ \beta, I]}", color={rgb,255:red,128;green,128;blue,128}, squiggly, from=2-1, to=3-6]
    \arrow["{[\rho^{-1} \circ \phi, I] \circ [\rho^{-1} \circ \beta, I]}"', color={rgb,255:red,128;green,128;blue,128}, squiggly, from=2-1, to=1-6]
  \end{tikzcd}\]
  Similarly, tensoring coherences $\beta$ and $\phi$ in $\C$ and then lifting is
  equivalent to first lifting them individually and then tensoring them in 
  $\Lb(\C)$ by
  % https://q.uiver.app/?q=WzAsNyxbMCwxLCJBIFxcb3RpbWVzIEEnIl0sWzIsMCwiKEIgXFxvdGltZXMgSSkgXFxvdGltZXMgKEInIFxcb3RpbWVzIEkpIl0sWzMsMCwiKEIgXFxvdGltZXMgQicpIFxcb3RpbWVzIChJIFxcb3RpbWVzIEkpIl0sWzEsMCwiQiBcXG90aW1lcyBCJyJdLFsxLDIsIkIgXFxvdGltZXMgQiciXSxbMywyLCIoQiBcXG90aW1lcyBCJynCoFxcb3RpbWVzIEkiXSxbMywxLCIoQiBcXG90aW1lcyBCJykgXFxvdGltZXMgSSJdLFswLDQsIlxcYmV0YVxcb3RpbWVzXFxwaGkiLDJdLFs0LDUsIlxccmhvXnstMX0iLDJdLFs1LDYsIlxcaWQgXFxvdGltZXMgXFxpZCIsMl0sWzIsNiwiXFxpZCBcXG90aW1lcyBcXHJobyJdLFswLDMsIlxcYmV0YVxcb3RpbWVzXFxwaGkiXSxbMywxLCJcXHJob157LTF9wqBcXG90aW1lcyBcXHJob157LTF9Il0sWzEsMiwiXFx2YXJ0aGV0YSJdLFswLDIsIltcXHJob157LTF9IFxcY2lyYyBcXGJldGEsSV0gXFxvdGltZXMgW1xccmhvXnstMX0gXFxjaXJjIFxccGhpLEldIiwyLHsiY29sb3VyIjpbMCwwLDUwXSwic3R5bGUiOnsiYm9keSI6eyJuYW1lIjoic3F1aWdnbHkifX19LFswLDAsNTAsMV1dLFswLDUsIltcXHJob157LTF9IFxcY2lyYyAoXFxiZXRhIFxcb3RpbWVzIFxccGhpKSxJXSIsMCx7ImNvbG91ciI6WzAsMCw1MF0sInN0eWxlIjp7ImJvZHkiOnsibmFtZSI6InNxdWlnZ2x5In19fSxbMCwwLDUwLDFdXV0=
  \[\begin{tikzcd}
  	& {B \otimes B'} & {(B \otimes I) \otimes (B' \otimes I)} & {(B \otimes B') \otimes (I \otimes I)} \\
  	{A \otimes A'} &&& {(B \otimes B') \otimes I} \\
  	& {B \otimes B'} && {(B \otimes B') \otimes I}
  	\arrow["\beta\otimes\phi"', from=2-1, to=3-2]
  	\arrow["{\rho^{-1}}"', from=3-2, to=3-4]
  	\arrow["{\id \otimes \id}"', from=3-4, to=2-4]
  	\arrow["{\id \otimes \rho}", from=1-4, to=2-4]
  	\arrow["\beta\otimes\phi", from=2-1, to=1-2]
  	\arrow["{\rho^{-1} \otimes \rho^{-1}}", from=1-2, to=1-3]
  	\arrow["\vartheta", from=1-3, to=1-4]
  	\arrow["{[\rho^{-1} \circ \beta,I] \otimes [\rho^{-1} \circ \phi,I]}"', color={rgb,255:red,128;green,128;blue,128}, squiggly, from=2-1, to=1-4]
  	\arrow["{[\rho^{-1} \circ (\beta \otimes \phi),I]}", color={rgb,255:red,128;green,128;blue,128}, squiggly, from=2-1, to=3-4]
  \end{tikzcd}\]
  In this way, coherence of the monoidal structure in $\Lb(\C)$ follows
  from that of \C. It remains to show is that the tensor product of morphisms is well-defined, and that it respects restrictions.
  
  Suppose that $[f,E] \sim [g,G]$ via mediators $E \tot{h_1} E_1 \fromt{h_2} \dots \fromt{h_n} G$ and intermediates $f_1, \dots, f_{n-1}$ with $\ridm{f} = \ridm{f_1} = \cdots = \ridm{f_{n-1}} = \ridm{g}$. Then
  \[
    \ridm{\vartheta \circ (f \otimes f')} = \ridm{f \otimes f'}
    = \ridm{f} \otimes \ridm{f'} = \ridm{g} \otimes \ridm{f'}
    = \ridm{g \otimes f'} = \ridm{\vartheta \circ (g \otimes f')}
  \]
  since $\vartheta$ is an isomorphism (and so total). Also $[f \otimes
  f', E \otimes E'] \sim [g \otimes f', G \otimes E']$:
  % https://q.uiver.app/?q=WzAsOSxbMSwwLCJBwqBcXG90aW1lcyBBJyJdLFswLDEsIihCIFxcb3RpbWVzIEUpIFxcb3RpbWVzIChCJyBcXG90aW1lcyBFJykiXSxbMCwyLCIoQiBcXG90aW1lcyBCJykgXFxvdGltZXMgKEUgXFxvdGltZXMgRScpIl0sWzEsMSwiKEIgXFxvdGltZXMgRV8xKSBcXG90aW1lcyAoQicgXFxvdGltZXMgRScpIl0sWzEsMiwiKEIgXFxvdGltZXMgQicpIFxcb3RpbWVzIChFXzEgXFxvdGltZXMgRScpIl0sWzMsMiwiKEIgXFxvdGltZXMgQicpIFxcb3RpbWVzIChHIFxcb3RpbWVzIEUnKSJdLFszLDEsIihCIFxcb3RpbWVzIEInKSBcXG90aW1lcyAoRyBcXG90aW1lcyBFJykiXSxbMiwxLCJcXGNkb3RzIl0sWzIsMiwiXFxjZG90cyJdLFsxLDIsIlxcdmFydGhldGEiLDJdLFswLDEsImYgXFxvdGltZXMgZiciLDJdLFszLDQsIlxcdmFydGhldGEiLDJdLFs2LDUsIlxcdmFydGhldGEiLDJdLFswLDYsImcgXFxvdGltZXMgZiciXSxbMiw0LCJcXGlkIFxcb3RpbWVzIChoXzEgXFxvdGltZXMgXFxpZCkiLDJdLFs1LDgsIlxcaWQgXFxvdGltZXMgKGhfbiBcXG90aW1lcyBcXGlkKSJdLFs2LDcsIihcXGlkIFxcb3RpbWVzIGhfbikgXFxvdGltZXMgXFxpZCJdLFsxLDMsIihcXGlkIFxcb3RpbWVzIGhfMSkgXFxvdGltZXMgXFxpZCIsMl0sWzcsMywiKFxcaWQgXFxvdGltZXMgaF8yKSBcXG90aW1lcyBcXGlkIl0sWzgsNCwiXFxpZCBcXG90aW1lcyAoaF8yIFxcb3RpbWVzIFxcaWQpIl0sWzAsMywiZl8xIFxcb3RpbWVzIGYnXFxxdWFkIFxcY2RvdHMiXV0=
  \[\begin{tikzcd}
    & {A \otimes A'} \\
    {(B \otimes E) \otimes (B' \otimes E')} & {(B \otimes E_1) \otimes (B' \otimes E')} & \cdots & {(B \otimes B') \otimes (G \otimes E')} \\
    {(B \otimes B') \otimes (E \otimes E')} & {(B \otimes B') \otimes (E_1 \otimes E')} & \cdots & {(B \otimes B') \otimes (G \otimes E')}
    \arrow["\vartheta"', from=2-1, to=3-1]
    \arrow["{f \otimes f'}"', from=1-2, to=2-1]
    \arrow["\vartheta"', from=2-2, to=3-2]
    \arrow["\vartheta"', from=2-4, to=3-4]
    \arrow["{g \otimes f'}", from=1-2, to=2-4]
    \arrow["{\id \otimes (h_1 \otimes \id)}"', from=3-1, to=3-2]
    \arrow["{\id \otimes (h_n \otimes \id)}", from=3-4, to=3-3]
    \arrow["{(\id \otimes h_n) \otimes \id}", from=2-4, to=2-3]
    \arrow["{(\id \otimes h_1) \otimes \id}"', from=2-1, to=2-2]
    \arrow["{(\id \otimes h_2) \otimes \id}", from=2-3, to=2-2]
    \arrow["{\id \otimes (h_2 \otimes \id)}", from=3-3, to=3-2]
    \arrow["{f_1 \otimes f'\quad \cdots}", from=1-2, to=2-2]
  \end{tikzcd}\]
  Similarly $[f' \otimes f, E' \otimes E] \sim [f' \otimes g, E'
  \otimes G]$. Finally,
  \[
    \ridm{[f,E] \otimes [f',E']} = 
   [\rho^{-1} \circ \ridm{\vartheta \circ (f \otimes f')}, I] =
   [\rho^{-1} \circ \ridm{f \otimes f'}, I] = [\rho^{-1} \circ (\ridm{f}
   \otimes \ridm{f'}), I]
  \]
  and the diagram below commutes:
   % https://q.uiver.app/?q=WzAsNyxbMCwxLCJBIFxcb3RpbWVzIEEnIl0sWzEsMiwiQSBcXG90aW1lcyBBJyJdLFszLDIsIihBIFxcb3RpbWVzIEEnKSBcXG90aW1lcyBJIl0sWzEsMCwiQSBcXG90aW1lcyBBJyJdLFsyLDAsIihBIFxcb3RpbWVzIEkpIFxcb3RpbWVzIChBJyBcXG90aW1lcyBJKSJdLFszLDAsIihBIFxcb3RpbWVzIEEnKSBcXG90aW1lcyAoSSBcXG90aW1lcyBJKSJdLFszLDEsIihBIFxcb3RpbWVzIEEnKSBcXG90aW1lcyBJIl0sWzAsMSwiXFxyaWRte2Z9wqBcXG90aW1lcyBcXHJpZG17Zid9IiwyXSxbMSwyLCJcXHJob157LTF9IiwyXSxbMCwzLCJcXHJpZG17Zn3CoFxcb3RpbWVzIFxccmlkbXtmJ30iXSxbMyw0LCJcXHJob157LTF9IFxcb3RpbWVzIFxccmhvXnstMX0iXSxbNCw1LCJcXHZhcnRoZXRhIl0sWzUsNiwiXFxpZCBcXG90aW1lcyBcXHJobyJdLFsyLDYsIlxcaWQgXFxvdGltZXMgXFxpZCIsMl0sWzAsNSwiXFxyaWRteyhmLEUpfSBcXG90aW1lcyBcXHJpZG17KGYnLEUnKX0iLDIseyJjb2xvdXIiOlswLDAsNTBdLCJzdHlsZSI6eyJib2R5Ijp7Im5hbWUiOiJzcXVpZ2dseSJ9fX0sWzAsMCw1MCwxXV0sWzAsMiwiXFxyaWRteyhmLEUpIFxcb3RpbWVzIChmJyxFJyl9IiwwLHsiY29sb3VyIjpbMCwwLDUwXSwic3R5bGUiOnsiYm9keSI6eyJuYW1lIjoic3F1aWdnbHkifX19LFswLDAsNTAsMV1dXQ==
   \[\begin{tikzcd}[row sep=5mm]
    & {A \otimes A'} & {(A \otimes I) \otimes (A' \otimes I)} & {(A \otimes A') \otimes (I \otimes I)} \\
    {A \otimes A'} &&& {(A \otimes A') \otimes I} \\
    & {A \otimes A'} && {(A \otimes A') \otimes I}
    \arrow["{\ridm{f} \otimes \ridm{f'}}"', from=2-1, to=3-2]
    \arrow["{\rho^{-1}}"', from=3-2, to=3-4]
    \arrow["{\ridm{f} \otimes \ridm{f'}}", from=2-1, to=1-2]
    \arrow["{\rho^{-1} \otimes \rho^{-1}}", from=1-2, to=1-3]
    \arrow["\vartheta", from=1-3, to=1-4]
    \arrow["{\id \otimes \rho}", from=1-4, to=2-4]
    \arrow["{\id \otimes \id}"', from=3-4, to=2-4]
    \arrow["{\ridm{[f,E]} \otimes \ridm{[f',E']}}"', color={rgb,255:red,128;green,128;blue,128}, squiggly, from=2-1, to=1-4]
    \arrow["{\ridm{[f,E] \otimes [f',E']}}", color={rgb,255:red,128;green,128;blue,128}, squiggly, from=2-1, to=3-4]
   \end{tikzcd}\]
   This shows that $\ridm{[f,E] \otimes [f',E']} = \ridm{[f,E]} \otimes \ridm{[f',E']}$.
\end{proof}

\end{document}